\documentclass[reqno]{amsart}

\usepackage[linesnumbered,ruled]{algorithm2e}
\SetKwInput{KwInput}{Input}                
\SetKwInput{KwOutput}{Output}              

\usepackage{lineno,hyperref}

\usepackage{booktabs} 
\usepackage{amssymb,latexsym,amsfonts,amsmath,amsthm}
\usepackage{graphicx}
\usepackage{paralist}
\usepackage{capt-of}
\usepackage{xcolor}
\usepackage{dsfont}
\usepackage{tikz}
\usepackage[export]{adjustbox}
\usepackage{tcolorbox}
\usepackage{chngcntr}
\usepackage[normalem]{ulem}

\usetikzlibrary{calc,shapes,arrows,intersections,automata,shapes,chains}
\usetikzlibrary{shapes.multipart}

\usepackage{eso-pic}
\usepackage{epsfig}
\usepackage{mathtools}
\usepackage{epstopdf}
\usepackage{graphicx}
\usepackage{epsfig}
\usepackage{mathtools}
\usepackage{lineno}
\usetikzlibrary{positioning,shapes}
\usetikzlibrary{arrows}
\usepackage[justification=justified]{caption}
\allowdisplaybreaks

\newcommand{\tup}{\textup}

\newtheorem{theorem}{Theorem}
\newtheorem{definition}[theorem]{Definition}
\newtheorem{corollary}[theorem]{Corollary}

\newtheorem{assumption}[theorem]{Assumption}
\newtheorem{proposition}[theorem]{Proposition}
\newtheorem{remark}[theorem]{Remark}

\newcommand{\I}{\mathcal{I}_d}
\newcommand{\TF}{\mathcal{F}}

\newcommand{\op}{{\mathcal H}}
\def\field#1{\mathbb #1}%
\def\R{\field{R}}%
\def\N{\field{N}}%

\newcommand{\SF}{\mathcal{S}}
\newcommand{\Let}{:=}
\newcommand{\KK}{\mathcal{K}_{\infty}}
\newcommand{\intcc}[1]{\ensuremath{{\left[#1\right]}}}

\usepackage{color}

\usepackage{xspace}

\usepackage{enumerate}

\makeatletter
\long\def\@maketablecaption#1#2{\@tablecaptionsize
	\global \@minipagefalse
	\hbox to \hsize{\parbox[t]{\hsize}{\centering #1 \\ #2}}}
\makeatother

\begin{document}
	
	\begin{abstract}
	The security in information-flow has become a major concern for cyber-physical systems (CPSs). In this work, we focus on the analysis of an information-flow security property, called \emph{opacity}. Opacity characterizes the plausible deniability of a system's secret in the presence of a malicious outside intruder. We propose a methodology of checking a notion of opacity, called approximate initial-state opacity, for networks of discrete-time switched systems. Our framework relies on compositional constructions of finite abstractions for networks of switched systems and their so-called approximate initial-state opacity-preserving simulation functions (InitSOPSFs). Those functions characterize how close concrete networks and their finite abstractions are in terms of the satisfaction of approximate initial-state opacity. We show that such InitSOPSFs can be obtained compositionally by assuming some small-gain type conditions and composing so-called local InitSOPSFs constructed for each subsystem separately. Additionally, assuming certain stability property of switched systems, we also provide a technique on constructing their finite abstractions together with the corresponding local InitSOPSFs. Finally, we illustrate the effectiveness of our results through an example.
	\end{abstract}
	
	\title[Verification of Initial-State Opacity for Switched Systems: A Compositional Approach]{Verification of Initial-State Opacity for Switched Systems: A Compositional Approach}

	\author{Siyuan Liu$^1$}
	\author{Abdalla Swikir$^2$}
	\author{Majid Zamani$^{3,4}$}
	\address{$^1$Electrical and Computer Engineering Department, Technical University of Munich, Germany.}
	\email{sy.liu@tum.de}
	\address{$^2$Electrical and Computer Engineering Department, Technical University of Munich, Germany.}
	\email{abdalla.swikir@tum.de}
	\address{$^3$Computer Science Department, University of Colorado Boulder, USA.}
	\email{majid.zamani@colorado.edu}
	\address{$^4$Computer Science Department, Ludwig Maximilian University of Munich, Germany.}
	\maketitle
	
	\section{Introduction}
Cyber-physical systems (CPSs) are complex systems resulting from intricate interaction between embedded cyber devices and physical plants. In recent decade, CPSs have become ubiquitous in critical infrastructures and industrial control systems, including power plants, medical devices and smart communities \cite{cardenas2009challenges}. While the increased interaction between cyber and physical components increases systems' functionalities, it also exposes CPSs to more vulnerabilities and security challenges. Recently, the world has witnessed numerous cyber-attacks which have led to great losses in people's livelihoods \cite{ashibani2017cyber}. 
Therefore, ensuring the security of CPSs has become significantly more important.  

In this work, we focus on an information-flow security property, called \emph{opacity}, which characterizes the ability that a system forbids leaking its secret information to a malicious intruder outside the system. Opacity was first introduced in \cite{mazare2004using} to analyze cryptographic protocols. Later, opacity was widely studied in the domain of Discrete Event Systems (DESs), see \cite{lafortune2018history} and the references therein. In this context, existing works on the analysis of various notions of opacity mostly apply to systems modeled by finite state automata, which are more suitable for the cyber-layers of CPSs. However, for the physical components, system dynamics are in general hybrid with uncountable number of states. 

\subsection{Related Works}There have been some recent attempts to extend the notion of opacity to continuous-space dynamical systems \cite{ramasubramanian2019notions,zhang2019opacity, yin2019approximate, liu2020stochastic}. In \cite{ramasubramanian2019notions}, a framework for opacity was introduced for the class of discrete-time linear systems, where the notion of opacity was formulated as an output reachability property rather than an information-flow one. 
The results in \cite{zhang2019opacity} presented a formulation of opacity-preserving (bi)simulation relations between transition systems, which allows one to verify opacity of an infinite-state transition system by leveraging its associated finite quotient one. However, the notion of opacity proposed in this work assumes that the outputs of systems are symbols and are exactly distinguishable from each other, thus, is only suitable for systems with purely logical output sets. In a more recent paper \cite{yin2019approximate}, a new notion of \emph{approximate opacity} was proposed to accommodate imperfect measurement precision of intruders. Based on this, the authors proposed a notion of so-called approximate opacity-preserving simulation relation to capture the closeness between continuous-space systems and their finite abstractions (a.k.a symbolic models) in terms of preservation of approximate opacity. 
The recent results in \cite{liu2020stochastic} investigated opacity for discrete-time stochastic control systems using a notion of so-called initial-state opacity-preserving stochastic simulation functions between stochastic control systems and their finite abstractions (a.k.a. finite Markov decision processes). 

Although the results in \cite{zhang2019opacity,yin2019approximate,liu2020stochastic} look promising, the computational complexity of the construction of finite abstractions in those works grows exponentially with respect to the dimension of the state set, and, hence, those existing approaches will become computationally intractable when dealing with large-scale systems.

Motivated by those abstraction-based techniques in \cite{zhang2019opacity,yin2019approximate,liu2020stochastic} and their limitations, this work proposes an approach to analyze approximate initial-state opacity for networks of switched systems by constructing their opacity-preserving finite abstractions compositionally. There have been some recent results proposing compositional techniques for constructing finite abstractions for networks of systems (see the results in \cite{meyer,7403879,SWIKIR2019,Mallik19,swikir2019compositional} for more details). However, the aforementioned compositional schemes are proposed for the sake of controller synthesis for temporal logic properties, and none of them are applicable to deal with security properties including opacity. 

\subsection{Contributions}In this paper, we provide for the first time a compositional approach to analyze approximate initial-state opacity of a network of switched systems using their finite abstractions. 
A new notion of so-called approximate initial-state opacity-preserving simulation function (InitSOPSF) is introduced as a system relation to characterize the closeness between two networks in terms of preservation of approximate initial-state opacity. We show that such an InitSOPSF can be established by composing certain local InitSOPSFs which relate each switched subsystem to its local finite abstraction. Moreover, under some assumptions ensuring incremental input-to-state stability of discrete-time switched systems, an approach is provided to construct local finite abstractions along with the corresponding local InitSOPSFs for all of the subsystems. 
Then, we derive some small-gain type conditions, under which one can construct a finite abstraction of the concrete network of switched systems by interconnecting local finite abstractions of subsystems. Finally, one can leverage the constructed finite abstraction of the network to check its opacity. The proposed compositional abstraction-based opacity verification pipeline is depicted in Figure~\ref{pipeline}.
\begin{figure}[ht!]
	\centering
	\includegraphics[width=0.8\textwidth]{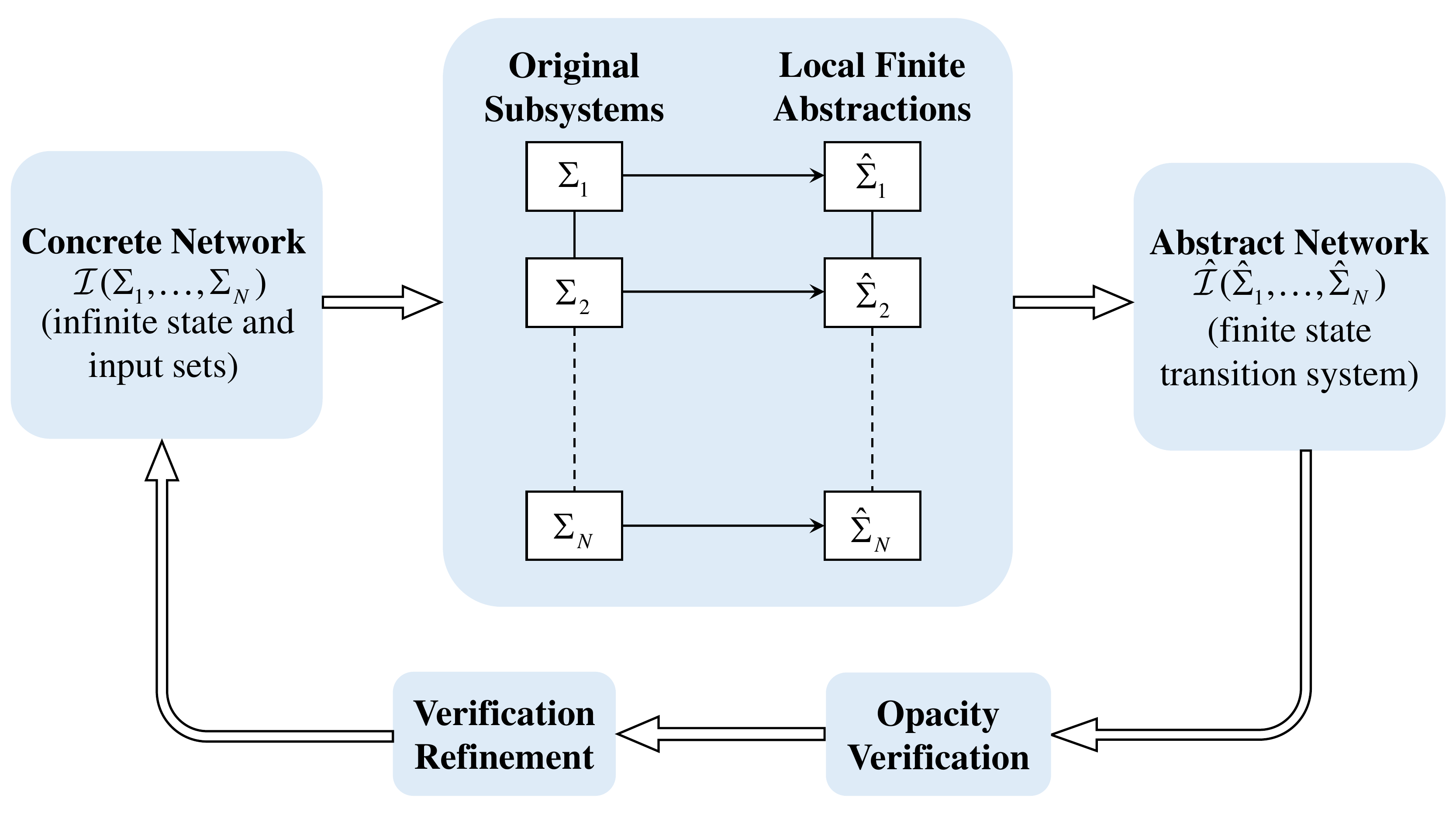}
	\caption{Compositional framework of opacity verification for networks of switched systems.}
	\label{pipeline}
\end{figure}

\subsection{Organization} The rest of this paper is organized as follows. In Section~\ref{1:II}, we first introduce necessary notations and preliminaries of the paper. Then, a new notion of approximate opacity preserving simulation functions (InitSOPSFs) is proposed in Section~\ref{sec:InitSOPSF}. In Section~\ref{sec:compoabs}, we provide a compositional framework for the construction of InitSOPSF for a network of discrete-time switched systems. In Section~\ref{1:IV}, we present how to construct local finite abstractions for a class of incrementally input-to-state stable subsystems, and then propose a small-gain type condition required for the main compositionality result. 
Next, an illustrative example is provided in Section~\ref{sec:example} that showcases how one can leverage our compositionality results for the verification of opacity for a network of switched systems. Finally, we conclude the paper in Section~\ref{sec:conclusion}.

	\section{Notation and Preliminaries}\label{1:II}
\subsection{Notation}
We denote by $\R$ and $\N$ the set of real numbers and non-negative integers,  respectively.
These symbols are annotated with subscripts to restrict them in
the obvious way, e.g. $\R_{>0}$ denotes the positive real numbers. We denote the closed, open, and half-open intervals in $\R$ by $[a,b]$,
$(a,b)$, $[a,b)$, and $(a,b]$, respectively. For $a,b\in\N$ and $a\le b$, we
use $[a;b]$, $(a;b)$, $[a;b)$, and $(a;b]$ to
denote the corresponding intervals in $\N$.  Given any $a\in\R$, $\vert a\vert$ denotes the absolute value of $a$.
Given $N\in\N_{\ge1}$ vectors $\nu_i\in\R^{n_i}$, $n_i\in\N_{\ge1}$, and $i\in[1;N]$, we
use $\nu=[\nu_1;\ldots;\nu_N]$ to denote the vector in $\R^n$ with
$n=\sum_i n_i$ consisting of the concatenation of vectors~$\nu_i$. Moreover, $\Vert \nu\Vert$ denotes the infinity norm of $\nu$.
The individual elements in a matrix $A\in \R^{m\times n}$, are denoted by $\{A\}_{i,j}$, where $i\in[1;m]$ and $j\in[1;n]$. 
We denote by $\text{card}(\cdot)$ the cardinality of a given set and by $\emptyset$ the empty set. 
For any set \mbox{$S\subseteq\R^n$} of the form of finite union of boxes, e.g., $S=\bigcup_{j=1}^MS_j$ for some $M\in\N$, where $S_j=\prod_{i=1}^{n} [c_i^j,d_i^j]\subseteq \R^{n}$ with $c^j_i<d^j_i$, we define $\emph{span}(S)=\min_{j=1,\ldots,M}\eta_{S_j}$ and $\eta_{S_j}=\min\{|d_1^j-c_1^j|,\ldots,|d_{n}^j-c_{n}^j|\}$. 
Moreover, for a set in the form of $X= \prod_{i=1}^N X_i$, where $X_i \subseteq \R^{n_i}$, $\forall i\in[1;N]$, are of the form of finite union of boxes, and any positive (component-wise) vector $\phi = [\phi_1;\dots;\phi_N]$ with $\phi_i \leq \emph{span}(X_i)$, $\forall i\in [1;N]$, we define $[X]_\phi= \prod_{i=1}^N [X_i]_{\phi_i}$, where $[X_i]_{\phi_i} = [\R^{n_i}]_{\phi_i}\cap{X_i}$ and  $[\R^{n_i}]_{\phi_i}=\{a\in \R^{n_i}\mid a_{j}=k_{j}\phi_i,k_{j}\in\mathbb{Z},j=1,\ldots,n_i\}$.
Note that if $\phi = [\eta;\dots;\eta]$, where $0<\eta\leq\emph{span}(S)$, we simply use notation $[S]_{\eta}$ rather than $[S]_{\phi}$. 
With a slight abuse of notation, we write $[S]_{0}:=S$. Note that $[S]_{\eta}\neq\emptyset$ for any $0\leq\eta\leq\emph{span}(S)$.
We use notations $\mathcal{K}$ and $\mathcal{K}_\infty$
to denote different classes of comparison functions, as follows:
$\mathcal{K}=\{\alpha:\mathbb{R}_{\geq 0} \rightarrow \mathbb{R}_{\geq 0}|$ $\alpha$ is continuous, strictly increasing, and $\alpha(0)=0\}$; $\mathcal{K}_\infty=\{\alpha \in \mathcal{K}|$ $ \lim\limits_{r \rightarrow \infty} \alpha(r)=\infty\}$.
For $\alpha,\gamma \in \mathcal{K}_{\infty}$ we write $\alpha\le\gamma$ if $\alpha(r)\le\gamma(r)$, and, with abuse of the notation, $\alpha=c$ if $\alpha(r)=cr$ for all $c,r\geq0$. Finally, we denote by $\I$ the identity function over $\R_{\ge0}$, that is $\I(r)=r, \forall r\in \R_{\ge0}$.
Given sets $X$ and $Y$ with $X\subset Y$, the complement of $X$ with respect to $Y$ is defined as $Y \backslash X = \{x : x \in Y, x \notin X\}.$ 

\subsection{Discrete-Time Switched Systems} 
We consider discrete-time switched systems of the following form.
\begin{definition}\label{dtss}
	A discrete-time switched system (dt-SS) $\Sigma$ is defined by the tuple $\Sigma=(\mathbb X,\mathbb X_0,\mathbb X_s, P,\mathbb W,F,\mathbb Y,h)$,
	where 
	\begin{itemize}
		\item $\mathbb X\subseteq\R^n$ is the state set;
		\item $\mathbb X_0\subseteq\R^n$ is the initial state set;
		\item $\mathbb X_s\subseteq\R^n$ is the secret state set;
		\item $P=\{1,\dots,m\}$ is the finite set of modes;
		\item $\mathbb W\subseteq\R^m$ is the internal input set;
		\item $F=\{f_1,\dots,f_m\}$ is a collection of set-valued maps $f_p: \mathbb X\times \mathbb W\rightrightarrows\mathbb X $ for all $p\in P$;
		\item $\mathbb Y\subseteq\R^q$ is the output set;
		\item $h: \mathbb X \rightarrow \mathbb Y $ is the output map.
	\end{itemize} 
	The dt-SS $\Sigma $ is described by difference inclusions of the form
	\begin{align}\label{eq:2}
		\Sigma:\left\{
		\begin{array}{rl}
			\mathbf{x}(k+1)&\in f_{\mathsf{p}(k)}(\mathbf{x}(k),\omega(k)),\\
			\mathbf{y}(k)&=h(\mathbf{x}(k)),
		\end{array}
		\right.
	\end{align}where $\mathbf{x}:\mathbb{N}\rightarrow \mathbb X $, $\mathbf{y}:\mathbb{N}\rightarrow \mathbb Y$, $\mathsf{p}:\mathbb{N}\rightarrow P$, and $\omega:\mathbb{N}\rightarrow \mathbb W$ are the state, output, switching, and internal input signal, respectively.
\end{definition}
Let $\varphi_k, k \in \N_{\ge1}$, denote the time when the $k$-th switching instant occurs. 
We assume that signal $\mathsf{p}$ satisfies a dwell-time condition \cite{liberzon} (i.e. there exists $k_d \in \N_{\ge1}$, called the dwell-time, such that for all consecutive switching time instants $\varphi_k,\varphi_{k+1}$, $\varphi_{k+1}-\varphi_{k}\geq k_d$).	
If for all $x\in  \mathbb X, p\in  P,  w \in  \mathbb W$, $\text{card}(f_{p}(x,w))\leq1$, we say the system $\Sigma$ is deterministic, and non-deterministic otherwise. System $\Sigma$ is called finite if $ \mathbb X, \mathbb W$ are finite sets and infinite otherwise. Furthermore, if for all $x\in  \mathbb X$ 
there exist $ p\in  P$ and $ w \in  W $ such that $\text{card}(f_p(x,w))\neq0$ we say the system is non-blocking. In this paper, 	we only deal with non-blocking systems.

\subsection{Transition systems}\label{subsec:transys}
In this subsection, we employ the notion of transition systems, introduced in \cite{swikirecc}, to provide an alternative description of switched systems that can be later directly related to their finite abstractions in a common framework. 
\begin{definition}\label{tsm} Given a dt-SS $\Sigma=(\mathbb X,\mathbb X_0,\mathbb X_s,P,\mathbb W,F,\mathbb Y,h)$, we define the associated transition system $T(\Sigma)=(X,X_0,X_s,U,W,\TF,Y,\op)$, where:
	\begin{itemize}
		\item  $X=\mathbb X\times P\times \{0,\dots,k_d-1\}$ is the state set; 
		\item  $X_0=\mathbb X_0\times P\times \{0\}$ is the initial state set; 
		\item  $X_s=\mathbb X_s\times P\times \{0,\dots,k_d-1\}$  is the secret state set; 
		\item $U=P$ is the external input set;
		\item $W=\mathbb{W}$ is the internal input set;
		\item $\TF$ is the transition function given by $(x^+,p^+,l^+)\in \TF((x,p,l),u,w)$ if and only if  $x^+\in f_p(x,w),u=p$ and the following scenarios hold:
		\begin{itemize} 
			\item$l<k_d-1$, $p^+=p$ and $l^+=l+1$: switching is not allowed because the time elapsed since
			the latest switch is strictly smaller than the dwell time;
			\item $l=k_d-1$, $p^+=p$ and $l^+=k_d-1$: switching is allowed but no switch occurs;
			\item $l=k_d-1$, $p^+\neq p$ and $l^+=0$: switching
			is allowed and a switch occurs;
		\end{itemize}
		\item $Y=\mathbb{Y}$ is the output set;
		\item $\mathcal{H}:X\rightarrow Y$ is the output map defined as $\mathcal{H}(x,p,l)=h(x)$.
	\end{itemize}
\end{definition}
Note that in the above definition, two additional variables $p$ and $l$ are added to the state tuple of the system $\Sigma$. The variable $p$ captures whether or not a switching is allowed for the system at a given time instant, and $l$ serves as a memory to record the sojourn of switching signal.

The following proposition is borrowed from \cite{swikir2019compositional} showing that the output runs of a dt-SS $\Sigma$ and its associated transition system $T(\Sigma)$ are equivalent so that one can use $\Sigma$ and $T(\Sigma)$ interchangeably.
\begin{proposition}\label{traj}
	Consider a transition system $T(\Sigma)$ in Definition \ref{tsm} associated to $\Sigma$ as in Definition \ref{dtss}. Any output trajectory of $\Sigma$ can be uniquely equated to an output trajectory of $T(\Sigma)$ and vice versa.
\end{proposition}	

Next, let us introduce a formal definition of networks of dt-SS (or equivalently, networks of transition systems).

\subsection{Networks of Systems}
Consider $N\in\N_{\ge1}$ dt-SS $\Sigma_i=(\mathbb X_i, \mathbb X_{0_i},\mathbb X_{s_i}, P_i,\mathbb W_i,$ $F_i,\mathbb Y_i,h_i)$, $i\in[1;N]$, with partitioned internal inputs and outputs as
\begin{align}\label{eq:int1}
	w_i=[w_{i1};\ldots;w_{i(i-1)};w_{i(i+1)};\ldots;w_{iN}],& \quad
	\mathbb{W}_i=\prod_{j=1, j\neq i}^{N} \mathbb{W}_{ij},\\\label{eq:int2}
	h_{i}(x_i)=[h_{i1}(x_i);\ldots; h_{iN}(x_i)],& \quad
	\mathbb Y_i=\prod_{j=1}^N  \mathbb Y_{ij},
\end{align}with $w_{ij} \in \mathbb{W}_{ij}$, and $y_{ij} =h_{ij}(x_i)\in \mathbb Y_{ij}$.
The outputs $y_{ii}$ are considered as external ones, whereas $y_{ij}$  with $i\neq j$ are interpreted as internal ones which are used to construct interconnections between systems. In particular, we assume that $w_{ij}=y_{ji}$, if there is connection from system $\Sigma_{j}$ to
$\Sigma_i$, otherwise, we set $h_{ji}\equiv 0$. In the sequel, we denote by $\mathcal{N}_i = \{j \in[1;N], j\neq i|h_{ji}\neq 0\}$ the collection of neighboring systems $\Sigma_j,j\in\mathcal{N}_i$, that provide internal inputs to system $\Sigma_i$.

Now, we are ready to provide a formal definition of the concrete network consisting of $N\in\N_{\ge1}$ dt-SS.
\begin{definition}\label{netsw}
	Consider $N\in\N_{\ge1}$ dt-SS $\Sigma_i=(\mathbb X_i, \mathbb X_{0_i},\mathbb X_{s_i}, P_i,\mathbb W_i,$ $F_i,\mathbb Y_i,h_i)$, $i\in[1;N]$ with the input-output structure given by $\eqref{eq:int1}$ and $\eqref{eq:int2}$. The network, representing the interconnection of  $N\in\N_{\ge1}$ dt-SS $\Sigma_i$, 
	is a tuple $\Sigma=(\mathbb X,\mathbb X_0,\mathbb X_s,P,F,\mathbb Y,h)$, denoted by $\mathcal{I}(\Sigma_1,\ldots,\Sigma_N)$, where $\mathbb X =\prod_{i=1}^N \mathbb X_i$, $\mathbb X_0 =\prod_{i=1}^N \mathbb X_{0_i}$, $\mathbb X_s =\prod_{i=1}^N \mathbb X_{s_i}$,
	$P=\prod_{i=1}^N  P_i$, ${F}=\prod_{i=1}^N{F}_i$, $ \mathbb Y=\prod_{i=1}^N  \mathbb Y_{ii}$, 
	$h(x)\Let \intcc{h_{11}(x_1);\ldots;h_{NN}(x_N)}$ with $x=\intcc{x_{1};\ldots;x_{N}}$, subject to the constraint:
	\begin{align}\label{const}
		y_{ji}= w_{ij}, \mathbb{Y}_{ji} \subseteq \mathbb{W}_{ij}, \forall i\in [1;N], j\in\mathcal{N}_i. 
	\end{align}
\end{definition} 

Similarly, given transition systems $T(\Sigma_i)$, one can also define a network of transition systems $\mathcal{I}(T(\Sigma_1),\!\ldots\!,T(\Sigma_N))$. For the rest of the paper, we mainly deal with the transition systems as they allow us to model dt-SS $\Sigma$ and their finite abstractions in a common framework.

For an interconnection of $N\in\N_{\ge1}$ finite dt-SS $\hat \Sigma_i$, with input-output structure configuration as in  $\eqref{eq:int1}$ and $\eqref{eq:int2}$, we introduce the following definition of networks of finite dt-SS.
\begin{definition}\label{absnetsw}
	Consider $N\in\N_{\ge1}$ finite dt-SS $\hat \Sigma_i=( \hat {\mathbb X}_i, \hat {\mathbb X}_{0_i},\hat {\mathbb X}_{s_i}, \hat P_i,\hat {\mathbb W}_i,$ $ \hat F_i,\hat {\mathbb Y}_i,$ $\hat h_i)$, $i\in[1;N]$ with the input-output structure given by $\eqref{eq:int1}$ and $\eqref{eq:int2}$. The network, representing the interconnection of  $N\in\N_{\ge1}$ finite dt-SS $\hat \Sigma_i$, 
	is a tuple $\hat\Sigma=(\hat {\mathbb X},\hat {\mathbb X}_0,\hat {\mathbb X}_s,\hat P,\hat F,\hat {\mathbb Y},\hat h)$, denoted by $\hat {\mathcal{I}}(\hat \Sigma_1,\ldots,\hat \Sigma_N)$, where $\hat {\mathbb X} =\prod_{i=1}^N \hat {\mathbb X}_i$, $\hat {\mathbb X}_0 =\prod_{i=1}^N \hat {\mathbb X}_{0_i}$, $\hat {\mathbb X}_s =\prod_{i=1}^N \hat {\mathbb X}_{s_i}$,
	$\hat P=\prod_{i=1}^N  \hat P_i$, $\hat {F}=\prod_{i=1}^N \hat {F}_i$, $\hat {\mathbb Y}=\prod_{i=1}^N \hat {\mathbb Y}_{ii}$, 
	$\hat h(x)\Let \intcc{\hat h_{11}(\hat x_1);\ldots;\hat h_{NN}(\hat x_N)}$ with $\hat x=\intcc{\hat x_{1};\ldots;\hat x_{N}}$, 	subject to the constraint:
	\begin{align}\label{absconst}
		\Vert \hat y_{ji}-\hat w_{ij} \Vert \leq \phi_{ij}, [\hat {\mathbb{Y}}_{ji}]_{\phi_{ij}} \subseteq \hat {\mathbb{W}}_{ij}, \forall i\in [1;N], j\in\mathcal{N}_i, 
	\end{align}
	where $\phi_{ij}$ is an internal input quantization parameter designed for constructing local finite abstractions (cf. Definition \ref{smm}).
\end{definition} 
Similarly, given finite transition systems $T(\hat \Sigma_i)$, one can also define a network of transition systems as $\hat{\mathcal{I}}(T(\hat \Sigma_1),\ldots,T(\hat \Sigma_N))$.

An example of a concrete network and an abstract network is illustrated in Figure \ref{system1}, where each consists of three switched subsystems.

\begin{figure}[h!]
	\centering
	\includegraphics[width=0.7\textwidth]{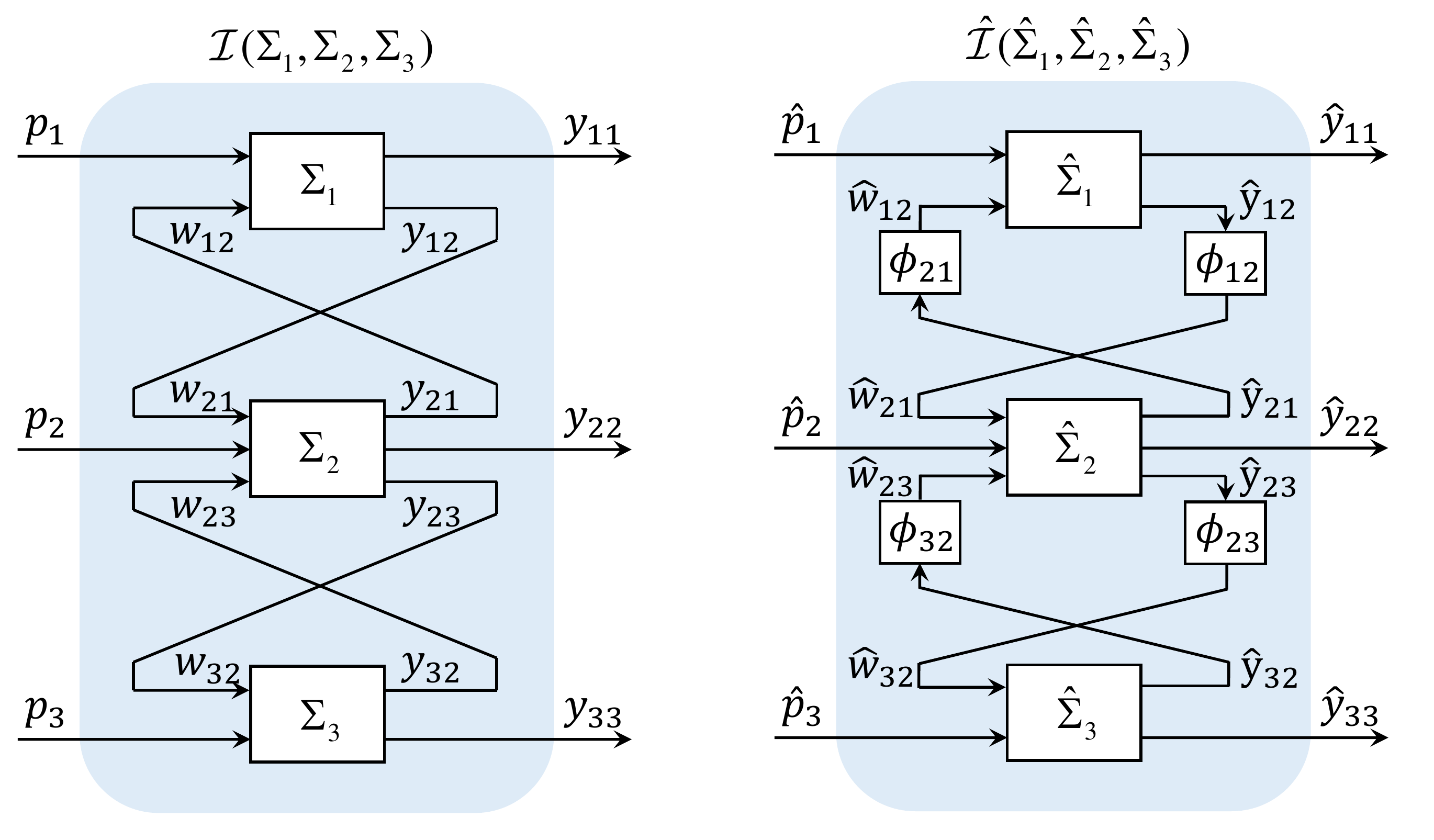}
	\caption{[Left]: Concrete network composed of three switched subsystems $\Sigma_1$, $\Sigma_2$, and $\Sigma_3$ with $h_{13}=h_{31}=0$, where $y_{ji}= w_{ij}$, $\forall i,j \in [1;3]$; [Right]: Abstract network composed of three finite subsystems $\hat \Sigma_1$, $\hat \Sigma_2$, and $\hat \Sigma_3$ with $\hat h_{13}= \hat h_{31}=0$, and the internal inputs $\hat w_{ij}$ for system $\hat \Sigma_i$ are taken from the discretized internal outputs of system $\hat \Sigma_j$ under the constraint $\Vert \hat y_{ji}-\hat w_{ij} \Vert \leq \phi_{ij}$, $\forall i,j \in [1;3]$, where $\phi_{ij}$ are internal input quantization parameters.}
	\label{system1}
\end{figure}

\begin{remark}
	Note that in the above definitions, the interconnection constraint in \eqref{const} for the concrete network is different from that for the abstract network in \eqref{absconst}. 
	For networks of finite abstractions,  due to possibly different granularities of finite internal input sets $\hat {\mathbb{W}}_{ij}$ and output sets $\hat {\mathbb{Y}}_{ij}$, we introduce parameters $\phi_{ij}$ in \eqref{absconst} for having a well-posed interconnection. The values of $\phi_{ij}$ will be designed later in Definition \ref{smm} when constructiong local finite abstractions of the subsystems.
	
\end{remark}

Before introducing the notion of approximate initial-state opacity for networks of transition systems, we introduce some notations that will be used to characterize opacity property. 
Consider network $T(\Sigma)$. We use $z^k$ to denote the state of $T(\Sigma)$ reached at time $k \in \mathbb{N}$ from initial state $z^0$ under an input sequence ${\bar u}$ with length $k$, and denote by $\{z^0, z^1, \dots, z^n\}$ a finite state run of $T(\Sigma)$ with length $n\in \mathbb{N}$.

\subsection{Approximate Initial-state Opacity}
Here, let us review a notion of approximate initial-state opacity \cite{yin2019approximate}. In this context, the system's behaviors are assumed to be observed by an outside intruder which aims at inferring secret information of the system. The adopted concept of secrets are formulated as state-based.

\begin{definition}
	Consider network $T(\Sigma)=(X,X_0,X_s,U,\TF,Y,\op)$  
	and a constant $\delta \geq 0$. Network $T(\Sigma)$ is said to be
	$\delta$-approximate initial-state opaque if for any $z^0 \in X_0 \cap X_s$ and finite state run $\{z^0, z^1, \dots, z^n\}$, there exist $\bar{z}^0 \in X_0 \setminus X_s$ and a finite state run $\{\bar{z}^{0}, \bar{z}^{1}, \dots, \bar{z}^{n}\}$ such that 		
	$$\max_{k \in [0;n]} \Vert \mathcal{H}(z^k)-\mathcal{H}(\bar{z}^k) \Vert \leq \delta.$$
\end{definition}

Intuitively, the notion of $\delta$-approximate initial-state opacity requires that, whenever observing any output run, an intruder with measurement precision $\delta$ is never certain that the system is initiated from a secret state. In other words, the systems' secret information can never be revealed in the presence of an intruder that does not have an enough measurement precision.

\begin{remark}
	The approximate initial-state opacity is, in general, hard to check for a concrete network since there is no systematic way in the literature to check opacity for systems with infinite state set so far. 
	On the other hand, existing tool DESUMA\footnote{Available at URL http://www.eecs.umich.edu/umdes/toolboxes.html.} and algorithms~\cite{yin2017new},~\cite{saboori2013verification},\cite[Sec. IV]{zhang2019opacity} in DESs literature can be leveraged to
	check \emph{exact} opacity for networks with finite state sets. Therefore, it would be more feasible to verify opacity for networks consisting of finite abstractions and then carry back the verification result to concrete ones, given a formal simulation relation between those networks. To this purpose, an opacity-preserving simulation relation  will be introduced in the next section which formally relate a network of transition systems and its finite abstraction.  
\end{remark}


\section{Opacity Preserving Simulation Functions}\label{sec:InitSOPSF}

In this section, we introduce a notion of approximate initial-state opacity-preserving simulation function to quantitatively relate two networks of transition systems in terms of preserving approximate initial-state opacity. Such a function can be constructed compositionally as shown in Section \ref{sec:compoabs}.

Let us first recall the definition of approximate initial-state opacity-preserving simulation relations which was originally proposed in \cite{yin2019approximate}.
\begin{definition}\label{def:InitSOP}
	Consider networks $T(\Sigma)=(X,X_0,X_s,U,\TF,Y,$ $\op)$ and $T(\hat{\Sigma})=(\hat{X},\hat X_0,\hat X_s,\hat{U},\hat{\TF},\hat{Y},\hat{\op})$ where $\hat Y\subseteq Y$. For $\hat{\varepsilon} \in \mathbb R_{\ge 0}$, a relation  $R\subseteq X\times \hat{X}$  is called an  $\hat{\varepsilon}$-approximate initial-state opacity-preserving simulation relation ($\hat{\varepsilon}$-InitSOP simulation relation) from ${T}(\Sigma)$ to $T(\hat{\Sigma})$ if 
	\begin{enumerate} 
		\item[1] (${a}$) $\forall z^0 \in {X}_0 \cap {X}_s$, $\exists \hat z^0 \in \hat {X}_0 \cap \hat {X}_s$, s.t.	$(z^0,\hat z^0) \in R$;\\
		(${b}$) $\forall \hat z^0 \in \hat {X}_0 \setminus \hat {X}_s$, $\exists z^0 \in {X}_0 \setminus {X}_s$, s.t. $(z^0,\hat z^0) \in R$;
		\item[2]  $\forall (z, \hat{z}) \in R$, $\Vert \mathcal{H}(z) - \mathcal{\hat H}(\hat z) \Vert\leq \hat{\varepsilon}$;
		\item[3]  For any $(z, \hat{z}) \in R$, one has:\\
		(${a}$) $\forall u\in  U$, $\forall z^+ \in\TF(z,u)$, $\exists \hat u\in \hat U$, $\exists \hat z^+ \in\hat{\TF}(\hat{z},\hat{u})$, s.t. $(z^+,\hat z^+) \in R $; \\
		(${b}$) $\forall \hat u\in \hat U$, $\forall \hat z^+ \in \hat{\TF}(\hat{z},\hat{u})$, $\exists u\in  U$, $\exists z^+ \in \TF(z,u)$, s.t. $(z^+,\hat z^+) \in R $.
	\end{enumerate}
\end{definition}
The following corollary borrowed from \cite{yin2019approximate} shows the usefulness of Definition \ref{def:InitSOP} in terms of preserving approximate opacity across related networks.
\begin{corollary}\label{thm:InitSOP}
	Consider networks $T(\Sigma)=(X,X_0,X_s,U,\TF,Y,$ $\op)$ and $T(\hat{\Sigma})=(\hat{X},\hat X_0,\hat X_s,\hat{U},\hat{\TF},\hat{Y},\hat{\op})$ where $\hat Y\subseteq Y$. Let $\hat\varepsilon,\delta\in\mathbb R_{\ge 0}$.
	If  there exists an $\hat{\varepsilon}$-InitSOP simulation relation from ${T}(\Sigma)$ to $T(\hat{\Sigma})$ as in Definition \ref{def:InitSOP} and $\hat{\varepsilon}\leq \frac{\delta}{2}$,
	then the following implication holds \vspace{-0.1cm}
	\begin{align}
		&T(\hat{\Sigma})\tup{ is ($\delta-2\hat{\varepsilon}$)-approximate initial-state opaque} \nonumber \\
		&\Rightarrow {T}(\Sigma) \tup{ is $\delta$-approximate initial-state opaque}.\nonumber
	\end{align}
\end{corollary}
The above implication across related networks provides us a sufficient condition for verifying approximate initial-state opacity of a complex network using abstraction-based techniques. Particularly, when confronted with a large network of switched systems, one can construct a finite abstraction  $T(\hat{\Sigma})$  of the concrete network $T(\Sigma)$, conduct the opacity verification over the simpler network $T(\hat{\Sigma})$ and carry back the results to the concrete one. However, the above-mentioned InitSOP simulation relation is in general difficult to establish, especially when one is interested to build such a relation in a compositional framework. Therefore, in the following, we introduce a new notion of approximate initial-state opacity-preserving simulation function that relates two networks in terms of preserving approximate initial-state opacity.

\begin{definition}\label{sfg}
	Consider networks $T(\Sigma)=(X,X_0,X_s,U,\TF,Y,\op)$ and $T(\hat{\Sigma})=(\hat{X},\hat X_0,\hat X_s,$ $\hat{U},\hat{\TF},\hat{Y},\hat{\op})$ with $\hat Y\subseteq Y$. For $\varepsilon \in \mathbb{R}_{\geq0}$, a function ${\SF}: X\times \hat X \to \mathbb{R}_{\geq0} $ is called an $\varepsilon$-approximate initial-state opacity-preserving simulation function ($\varepsilon$-InitSOPSF) from ${T}(\Sigma)$ to $T(\hat{\Sigma})$ if there exists a function $\alpha \in \mathcal{K_{\infty}}$ such that 
	\begin{enumerate}
		\item[1](a) $\forall z^0 \in {X}_0 \cap {X}_s$, $\exists \hat z^0 \in \hat {X}_0 \cap \hat {X}_s$, s.t. ${\SF}(z^0,\hat z^0) \leq \varepsilon $;\\
		(b) $\forall \hat z^0 \in \hat {X}_0 \setminus \hat {X}_s$, $\exists z^0 \in {X}_0 \setminus {X}_s$, s.t. ${\SF}(z^0,\hat z^0) \leq \varepsilon $;
		\item[2]   $\forall z \in X, \forall \hat z \in \hat {X}$, $\alpha(\Vert \mathcal{H}(z) - \mathcal{\hat H}(\hat z) \Vert) \leq {\SF}(z,\hat z)$;
		\item[3] $\forall z \in X, \forall \hat z \in \hat {X}$ s.t. $\mathcal{S}(z,\hat z) \leq \varepsilon$, one has: \\
		(a) $\forall u\in  U$, $\forall z^+\in\TF(z,u)$, $\exists \hat u\in \hat U$, $\exists \hat z^+ \in\hat{\TF}(\hat{z},\hat{u})$, s.t. ${\SF}(z^+,\hat z^+) \leq \varepsilon$;\\
		(b) $\forall \hat u\in \hat U$, $\forall \hat z^+ \in \hat{\TF}(\hat{z},\hat{u})$, $\exists u\in U$, $\exists z^+ \in \TF(z,u)$,  s.t. ${\SF}(z^+,\hat z^+) \leq \varepsilon$.
	\end{enumerate}
	We say that $T(\hat{\Sigma})$ is an abstraction of ${T}(\Sigma)$ if there exists an $\varepsilon$-InitSOPSF from ${T}(\Sigma)$ to $T(\hat{\Sigma})$. In addition, if $T(\hat{\Sigma})$ is finite ($\hat X$ is a finite set), system $T(\hat{\Sigma})$ is called a finite abstraction (symbolic model) of the network ${T}(\Sigma)$, and is denoted by ${T}(\Sigma) \preceq^{\varepsilon} T(\hat{\Sigma})$.
\end{definition}


Although Definition \ref{sfg} is general in the sense that networks $T(\Sigma)$ and $T(\hat{\Sigma})$ can be either infinite or finite, network $T(\hat{\Sigma})$ practically consists of $N\in\N_{\ge1}$ finite abstractions. Hence, checking approximate initial-state opacity for this network is decidable in comparison to network $T(\Sigma)$.

The next result shows that the existence of an $\varepsilon$-InitSOPSF as we proposed in Definition \ref{sfg} for networks of transition systems implies the existence of an $\hat{\varepsilon}$-InitSOP simulation relation as in Definition \ref{def:InitSOP}.

\begin{proposition}\label{error}
	Consider networks $T(\Sigma)=(X,X_0,X_s,U,\TF,$ $Y,\op)$  and $T(\hat{\Sigma})=(\hat{X},\hat X_0,\hat X_s,\hat{U},\hat{\TF},\hat{Y},\hat{\op})$ where $\hat Y\subseteq Y$. Assume ${\SF}$ is an $\varepsilon$-InitSOPSF from ${T}(\Sigma)$ to $T(\hat{\Sigma})$  as in Definition \ref{sfg}. Then, relation $R\subseteq X\times \hat{X}$ defined by \vspace{-0.2cm}
	\begin{align}\label{re}
		R=\left\{(z,\hat z)\in {X}\times \hat{X}|{\SF}(z,\hat z)\leq {\varepsilon}\right\}, \vspace{-0.3cm}
	\end{align}is an $\hat{\varepsilon}$-InitSOP simulation relation from ${T}(\Sigma)$ to $T(\hat{\Sigma})$ with  \vspace{-0.1cm}
	\begin{align}\label{er}
		\hat{\varepsilon}={\alpha}^{-1}({\varepsilon}).
	\end{align}\vspace{-0.5cm}
\end{proposition}
\begin{proof}
	Condition 1 follows immediately from condition 1 in Definition \ref{sfg}, i.e. $S(z^0,\hat z^0) \leq \varepsilon $. 
	Next, we show that  $\forall (z, \hat z) \in R$: $\Vert \mathcal{H}(z) - \mathcal{\hat H}(\hat z) \Vert \leq \hat{\varepsilon}$. From the definition of $R$ and condition 2 in Definition \ref{sfg}, it is readily seen that  $\Vert \mathcal{H}(z) - \mathcal{\hat H}(\hat z) \Vert \leq \alpha^{-1}(\varepsilon)= \hat{\varepsilon}$.
	Finally, we show condition 3 for $R$. Consider any pair of $(z,\hat{z})\in X\times \hat{X}$ in relation $R$ and by the definition of $R$, one has $S(z,\hat{z})\leq \varepsilon$. Additionally, from 3(a) in Definition \ref{sfg}, one also has $\forall u\in  U$, $\forall z^+ \in \TF(z,u)$, $\exists \hat u\in \hat U$, $\exists \hat z^+ \in \hat{\TF}(\hat{z},\hat{u})$ s.t. $S(z^+,\hat z^+) \leq \varepsilon$. Hence, it follows that $(z^+,\hat z^+) \in R$ which satisfies condition 3(${a}$) of $R$. Condition 3(${b}$) can be proved in the same way and is omitted here, which concludes the proof.  
\end{proof}

Given the results of Corollary \ref{thm:InitSOP} and Proposition \ref{error}, one can readily see that if there exists an $\varepsilon$-InitSOPSF from ${T}(\Sigma)$ to $T(\hat{\Sigma})$ as in Definition \ref{sfg} and $T(\hat{\Sigma})$ is ($\delta-2\hat{\varepsilon}$)-approximate initial-state opaque, then $T(\Sigma)$ is $\delta$-approximate initial-state opaque, where $\hat{\varepsilon}={\alpha}^{-1}({\varepsilon}) \leq  \frac{\delta}{2}$, and $\delta \in \mathbb R_{\ge 0}$.

\section{Compositional Construction of Approximate Initial-state Opacity Preserving Simulation Function} \label{sec:compoabs}
As shown in the previous section, the proposed $\varepsilon$-InitSOPSF can be used for checking approximate initial-state opacity of concrete networks by leveraging their finite abstractions.  
However, for a network consisting of a large number of switched systems, constructing the corresponding $\varepsilon$-InitSOPSF and the abstract network monolithically is not feasible in general due to curse of dimensionality. 
Hence, in this section, we introduce a compositional framework based on which one can break down the intricate task in parts that are more manageable to accomplish. 
In particular, we first relate local finite abstractions of the subsystems via so-called local InitSOPSFs. Then, one can obtain the abstract network by interconnecting the local finite abstractions of the subsystems. Additionally, the corresponding $\varepsilon$-InitSOPSF to capture the closeness between the concrete and the abstract networks can be established by composing the local InitSOPSFs as well.

Let us first introduce a notion of local InitSOPSF for switched subsystems with internal inputs in the following subsection.

\subsection{Local Approximate Initial-state Opacity Preserving Simulation Function}

Suppose that we are given $N$ dt-SS $\Sigma_i=(\mathbb X_i,\mathbb X_{0_i},\mathbb X_{s_i}, P_i,\mathbb W_i,F_i,\mathbb Y_i,h_i)$, $i\in[1;N]$, or equivalently, $T(\Sigma_i)=(X_i,X_{0_i},X_{s_i},U_i,W_i,\TF_i,$ $Y_i,\mathcal H_i)$. Moreover, we assume that each system $T(\Sigma_i)$ and its abstraction $T(\hat{\Sigma}_i)$ admit a local $\varepsilon_i$-InitSOPSF as defined next. 

\begin{definition}\label{sf}
	Consider transition systems $T(\Sigma_i)=(X_i,X_{0_i},X_{s_i},U_i,W_i,\TF_i,Y_i,$ $\op_i)$ and $T(\hat{\Sigma}_i)=(\hat{X}_i,\hat X_{0_i},\hat X_{s_i},$ $\hat{U}_i,\hat{W}_i,\hat{\TF}_i,\hat{Y}_i,\hat{\op}_i)$, for all $i\in [1;N]$,  where $\hat W_i\subseteq W_i$ and $\hat Y_i\subseteq Y_i$. For $\varepsilon_i \in \mathbb{R}_{\geq0}$, a function $\mathcal{S}_i : X_i\times \hat X_i \to \mathbb{R}_{\geq0} $ is called a local $\varepsilon_i$-InitSOPSF from $T(\Sigma_i)$ to $T(\hat{\Sigma}_i)$  if there exist a constant $\vartheta_i \in \mathbb R_{\ge 0}$, and a function $\alpha_i \in \mathcal{K_{\infty}}$ such that 
	\begin{enumerate}
		\item[1](a) $\forall z^0_i \in {X}_{0_i} \cap {X}_{s_i}$, $\exists \hat z^0_i \in \hat {X}_{0_i} \cap \hat {X}_{s_i}$, s.t. $\mathcal{S}_i(z^0_i,\hat z^0_i) \leq \varepsilon_i $;\\
		(b) $\forall \hat z^0_i\in \hat {X}_{0_i} \setminus \hat {X}_{s_i}$, $\exists z^0_i \in {X}_{0_i} \setminus {X}_{s_i}$, s.t. $\mathcal{S}_i(z^0_i,\hat z^0_i) \leq \varepsilon_i $;
		\item[2]   $\forall z_i \in X_i, \forall \hat z_i \in \hat {X}_i$, $\alpha_i(\Vert \mathcal{H}_i(z_i) - \mathcal{\hat H}_i(\hat z_i) \Vert) \leq \mathcal{S}_i(z_i,\hat z_i)$;
		\item[3] $\forall z_i \in X_i, \forall \hat z_i \in \hat {X}_i$ s.t. $\mathcal{S}_i(z_i,\hat z_i) \leq \varepsilon_i$, $\forall w_i \in W_i$, $\forall \hat w_i \in \hat{{W}_i}$ s.t. $\Vert w_i- \hat w_i \Vert \leq \vartheta_i$, one has: \\
		(a) $\forall u_i\in U_i$, $\forall z^+_i \in\TF_i(z_i,u_i,w_i)$, $\exists \hat u_i\in \hat U_i$, $\exists \hat z^+_i \in\hat{\TF}_i(\hat{z}_i,\hat{u}_i,\hat{w}_i)$ s.t. $\mathcal{S}_i(z^+_i,\hat z^+_i) \leq \varepsilon_i$;\\
		(b) $\forall \hat u_i\in \hat U_i$, $\forall \hat z^+_i \in \hat{\TF}_i(\hat{z}_i,\hat{u}_i,\hat{w}_i)$, $\exists u_i\in U_i$, $\exists z^+_i \in \TF_i(z_i,u_i,w_i)$ s.t. $\mathcal{S}_i(z^+_i,\hat z^+_i) \leq \varepsilon_i$.
	\end{enumerate}
	We say that $T(\hat{\Sigma}_i)$ is an abstraction of $T(\Sigma_i)$ if there exists a local $\varepsilon_i$-InitSOPSF from $T(\Sigma_i)$ to $T(\hat{\Sigma}_i)$. In addition, if $T(\hat{\Sigma}_i)$ is finite ($\hat X_i$ and $\hat W_i$ are finite sets), system $T(\hat{\Sigma}_i)$ is called a finite abstraction (symbolic model) of $T(\Sigma_i)$, and is denoted by $T(\Sigma_i) \preceq_{L}^{\varepsilon_i} T(\hat{\Sigma}_i)$
\end{definition}
Note that the local $\varepsilon_i$-InitSOPSFs are mainly proposed for constructing a $\varepsilon$-InitSOPSF for the networks and they are not directly used for deducing approximate initial-state opacity-preserving simulation relation. Next, we show how to compose the above defined local $\varepsilon_i$-InitSOPSFs so that it can be used to quantify the distance between two networks.  

\subsection{Compositional Construction of Initial-state Opacity Preserving Simulation Function} 
In this subsection, we provide one of the main results of the paper. 
The following theorem provides a compositional approach for the construction of an $\varepsilon$-InitSOPSF from $T(\Sigma)$ to $T(\hat{\Sigma})$ via local $\varepsilon_i$-InitSOPSFs from $T(\Sigma_i)$ to $T(\hat{\Sigma}_i)$.

\begin{theorem}\label{thm:3}
	Consider network $T(\Sigma)=\mathcal{I}(T(\Sigma_1),\ldots,T(\Sigma_{N}))$. Assume that each $T(\Sigma_i)$ admits an abstraction $T(\hat{\Sigma}_i)$ together with a local $\varepsilon_i$-InitSOPSF $\SF_i$, associated with function $\alpha_{i}$ and constant $\vartheta_i$ as in Definition \ref{sf}. 
	Let $\varepsilon = \max\limits_{i} \varepsilon_i$. 
	If $\forall i \in [1;N]$, $\forall j \in \mathcal{N}_i$, 
	\begin{align} \label{compoquaninit}
		\alpha^{-1}_{j}(\varepsilon_j) + \phi_{ij}\leq \vartheta_i,
	\end{align}where $\phi_{ij}$ is an internal input quantization parameter for constructing the local finite abstractions $T(\hat{\Sigma}_i)$,
	then, function ${\SF}:X\times \hat{X}\rightarrow \R_{\ge0}$ defined as
	\begin{align}\label{defVinit}
		{\SF}&(z,\hat{z})\Let\max\limits_{i}\{ \frac{\varepsilon}{\varepsilon_i} \SF_i(z_{i},\hat{z}_{i}) \},
	\end{align}is an $\varepsilon$-InitSOPSF from $T(\Sigma)=\mathcal{I}(T(\Sigma_1),\ldots,T(\Sigma_{N}))$ to $T(\hat{\Sigma})=\hat {\mathcal{I}}(T(\hat{\Sigma}_1),$ $\ldots,T(\hat{\Sigma}_{N}))$.
\end{theorem}

\begin{proof}
	First, we show that condition 1(a) in Definition \ref{sfg} holds. Consider any $z^0\in X_0\cap  X_{S}$. For any system  $T(\Sigma_i)$ and the corresponding $\varepsilon_i$-InitSOPSF $\SF_i$, from the definition of $\SF_i$, we have $\forall z^0_i \in X_{0_i} \cap X_{s_i}$, $\exists \hat z^0_i \in \hat{{X}}_{0_i} \cap \hat{{X}}_{s_i}$ s.t. 
	$\SF_i(z^0_i,\hat z^0_i) \leq \varepsilon_i$. Then, from the definition of ${\SF}$ in $\eqref{defVinit}$ we get ${\SF}(z^0, \hat z^0) \leq \varepsilon$, where $\hat z^0\in\hat{{X}}_0 \cap \hat{{X}}_{S}$. Thus, condition 1(a) in Definition \ref{sfg} holds. 
	Condition 1(b) can be proved in the same way, thus is omitted. 
	Now, we show that condition 2 in Definition \ref{sfg} holds for some $\mathcal{K}_{\infty}$ function ${\alpha}$.	
	Consider any $z=\intcc{z_1;\ldots;z_N}\in X$ and $\hat z=\intcc{\hat z_1;\ldots;\hat z_N}\in\hat{{X}}$. Then, using condition 2 in Definition \ref{sf}, one gets
	\begin{align*}\notag
		&\Vert  \op(z) - \hat{\op}(\hat{z})\Vert =\max\limits_{i}\{\Vert \op_{ii}(z_i) - \hat{\op}_{ii}(\hat{z}_i)\Vert\}\\
		&\leq\max\limits_{i}\{\Vert \op_{i}(z_i) - \hat{\op}_{i}(\hat{z}_i)\Vert\}
		\leq\max\limits_{i}\{\alpha^{-1}_i \circ  \SF_i(z_{i},\hat{z}_{i})\}\leq \hat{\alpha}\circ\max\limits_{i}\{ \frac{\varepsilon}{\varepsilon_i} \SF_i(z_{i},\hat{z}_{i}) \},
	\end{align*}where $\hat{\alpha}=\max\limits_{i}\{\alpha^{-1}_i\}$. By defining ${\alpha}= \hat{\alpha}^{-1}$, one obtains
	\begin{align*}\notag
		{\alpha}(\Vert \op(z)-\hat{\op}(\hat{z})\Vert )\leq {\SF}(z,\hat{z}),
	\end{align*}which satisfies condition 2 in Definition \ref{sfg}.
	Now, we show that condition 3 holds. Let us consider any 
	$z\in X$ and $\hat z\in\hat{{X}}$ such that ${\SF}(z,\hat z) \leq \varepsilon$. It can be seen that from the structure of ${\SF}$ in \eqref{defVinit}, we get $\SF_i(z_{i},\hat{z}_{i}) \leq \varepsilon_i$, $\forall i \in[1;N]$. For each pair of systems  $T(\Sigma_i)$ and $T(\hat \Sigma_i)$, the internal inputs satisfy the chain of inequality
	\begin{align*}
		&\Vert w_i- \hat{w}_i\Vert = \max\limits_{j \in \mathcal{N}_i}\{\Vert w_{ij}- \hat{w}_{ij}\Vert \} 
		= \max\limits_{j \in \mathcal{N}_i}\{\Vert y_{ji}-\hat{y}_{ji}+\hat{y}_{ji}-\hat{w}_{ij}\Vert \} \leq \max\limits_{j \in \mathcal{N}_i}\{\Vert y_{ji}-\hat{y}_{ji}\Vert + \phi_{ij}\} \\
		&\leq\max\limits_{j \in \mathcal{N}_i}\{\Vert \mathcal{H}_{j}(z_j)-\mathcal{\hat H}_{j}(\hat{z}_j)\Vert+ \phi_{ij} \}\leq\max\limits_{j \in \mathcal{N}_i}\{\alpha^{-1}_{j}\circ \SF_j(z_{j},\hat{z}_{j})+ \phi_{ij}\}
		\leq\max\limits_{j \in \mathcal{N}_i}\{\alpha^{-1}_{j}\circ \varepsilon_j+ \phi_{ij}\}.	
	\end{align*}Using \eqref{compoquaninit}, one has $\Vert w_i- \hat{w}_i\Vert \leq \vartheta_i$.
	Therefore, by Definition \ref{sf} for each pair of systems $T(\Sigma_i)$ and $T(\hat \Sigma_i)$, one has $\forall u_i \in  U_i$,  $\forall z^+_i \in \TF_i(z_i,u_i,w_i)$, there exists $\hat{u}_i\in \hat{ U}_i$ and $\hat z^+_i \in \hat{\TF}_i(\hat z_i,\hat u_i,\hat w_i)$ such that  $\SF_i(z^+_i,\hat z^+_i) \leq \varepsilon_i$. 
	As a result, we get $\forall u=\intcc{u_{1};\ldots;u_{N}}\in U$, $\forall z^+ \in \TF(z,u)$, there exists $\hat u=\intcc{\hat u_{1};\ldots;\hat u_{N}}\in\hat{{U}}$ and $\hat{z}^+ \in \hat{\TF}(\hat{z},\hat{u})$ such that ${\SF}(z^+, \hat{z}^+) \Let\max\limits_{i}\{ \frac{\varepsilon}{\varepsilon_i} \SF_i(z_{i}^+,\hat{z}_{i}^+) \} \leq \varepsilon$. Therefore, condition 3(a) in Definition \ref{sfg} is satisfied with $\varepsilon = \max\limits_{i} \varepsilon_i$. The proof of condition 3(b) uses the same reasoning as that of 3(a) and is omitted. Therefore, we conclude that $\mathcal{S}$ is a $\varepsilon$-InitSOPSF from  $T(\Sigma)=\mathcal{I}(T(\Sigma_1),\ldots,T(\Sigma_{N}))$ to $T(\hat{\Sigma})=\hat {\mathcal{I}}(T(\hat{\Sigma}_1),$ $\ldots,T(\hat{\Sigma}_{N}))$.
\end{proof}

Till here, we have seen that one can construct an abstraction of a network of switched systems by interconnecting local abstractions of the subsystems. The overall InitSOPSF between two networks is established by composing local InitSOPSFs as well. This abstract network safisties Definition \ref{sfg}, which allows us to check approximate opacity property over the simpler abstract network and carry the results back to the concrete network using the results provided in Corollary \ref{thm:InitSOP}.

Next, we are going to impose certain conditions on the dynamics of the subsystems, such that one can construct proper abstractions for all of the subsystems together with the corresponding local InitSOPSFs.

\section{Construction of Finite Abstractions}\label{1:IV}
In this section, we are going to explore how to construct finite abstractions together with local InitSOPSFs for subsystems. The dt-SS $\Sigma=(\mathbb X,\mathbb X_0,\mathbb X_s, P,$ $\mathbb W,F,\mathbb Y,h)$ are assumed to be infinite and deterministic. Moreover, we assume the output map $h$ satisfies the following general Lipschitz assumption: there exists an $\ell\in\KK$ such that: $\Vert h(x)-h(y)\Vert\leq \ell(\Vert x-y\Vert)$ for all $x,y\in \mathbb X$. Here, we also use $\Sigma_{p}$ to denote a dt-SS $\Sigma$ in \eqref{eq:2} with constant switching signal $\mathsf{p}(k)=p,~\forall k\in \N$.


\subsection{Construction of Local Finite Abstractions}
Note that throughout this subsection, we are mainly talking about switched subsystems rather than the overall network. However, for the sake of better readability, we omit index $i$ of subsystems throughout the text in this subsection, e.g., we write $T(\Sigma)$ instead of $T(\Sigma_i)$.

Here, we establish an $\varepsilon$-InitSOPSF between $T(\Sigma)$ and its finite abstraction by assuming that, for all $p \in P$, $\Sigma_p$ is incrementally input-to-state stable ($\delta$-ISS) \cite{ruffer} as defined next.
\begin{definition}\label{def:SFD1} 
	System $\Sigma_{p}$  is $\delta$-ISS if there exist functions $ V_p:\mathbb X\times \mathbb X \to \mathbb{R}_{\geq0} $, $\underline{\alpha}_p, \overline{\alpha}_p, \rho_{p} \in \mathcal{K}_{\infty}$, and constant $0<\kappa_p<1$, such that for all $x,\hat x\in \mathbb{X}$, and for all $w,\hat w\in \mathbb{W}$  
	\begin{align}\label{e:SFC11}
		\underline{\alpha}_p (\Vert x&-\hat{x}\Vert ) \leq V_p(x,\hat{x})\leq \overline{\alpha}_p (\Vert x-\hat{x}\Vert ),\\\label{e:SFC22}
		V_p(f_p(x,w),&f_p(\hat x,\hat w))
		\leq \kappa_p V_p(x,\hat{x})+\rho_{p}(\Vert w- \hat{w}\Vert ).
	\end{align}
\end{definition}

\begin{remark}\label{commonlya}
	We say that $V_p$, $\forall p\in P$, are multiple $\delta$-ISS Lyapunov functions for subsystem $\Sigma$ if it satisfies \eqref{e:SFC11} and \eqref{e:SFC22}. Moreover, if $V_p=V_{p^+}, \forall p,p^+\in P$, we omit the index $p$ in \eqref{e:SFC11}, \eqref{e:SFC22}, and say that $V$ is a common $\delta$-ISS Lyapunov function for system $\Sigma$. We refer interested readers to \cite{liberzon} for more details on common and multiple Lyapunov functions for switched systems.  
\end{remark}
Now, we show how to construct a local finite abstraction $T(\hat{\Sigma})$ of transition system $T(\Sigma)$ associated to the switched subsystem $\Sigma$ in which $\Sigma_{p}$  is $\delta$-ISS.
\begin{figure}[ht!]
	\centering
	\includegraphics[width=0.5\textwidth]{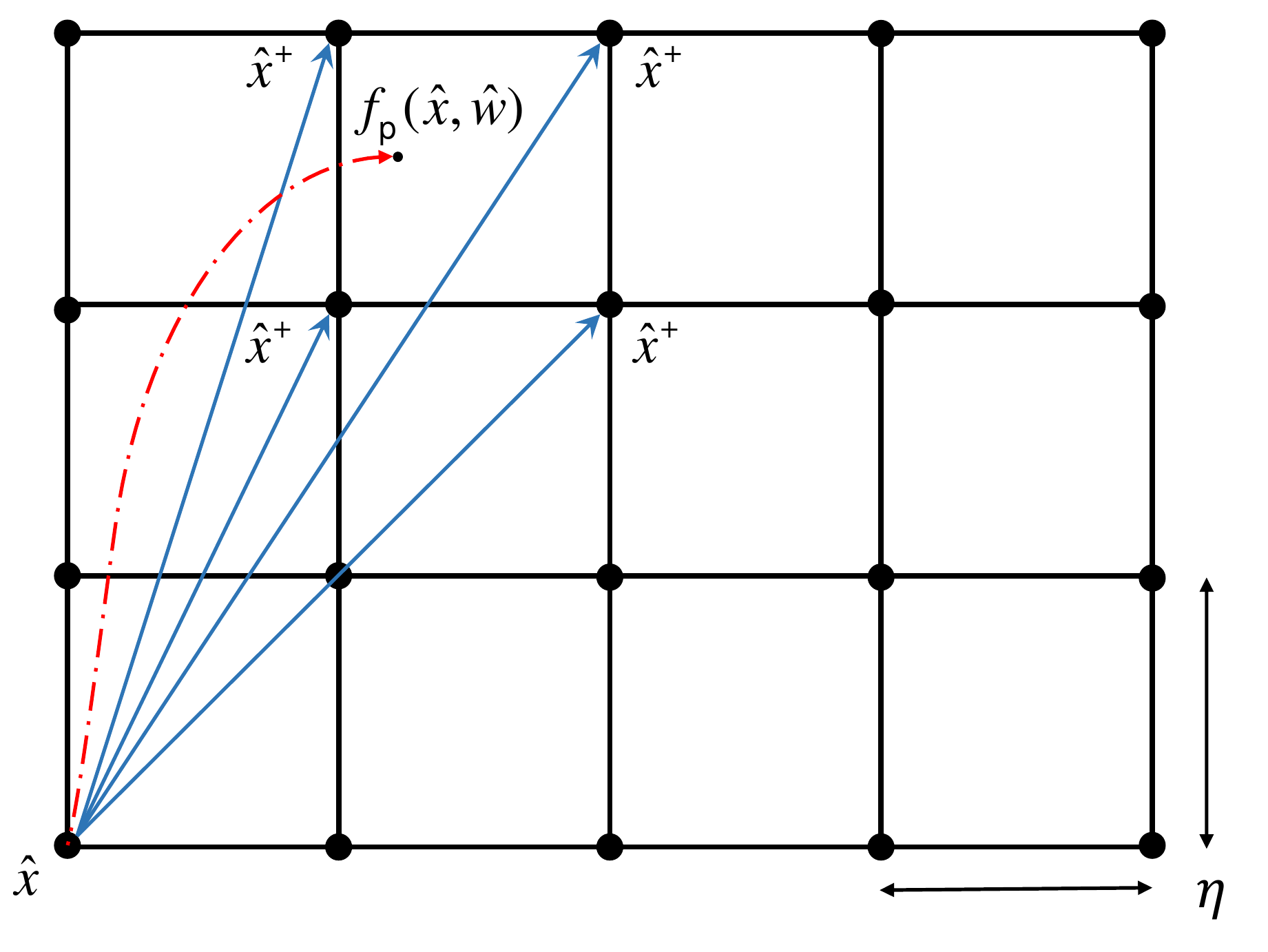}
	\caption{Construction of local finite abstractions of switched subsystems.}
	\label{abs}
\end{figure}

\begin{definition}\label{smm} Consider a transition system $T(\Sigma)=(X,X_0,X_s,U,W,\TF,Y,\op)$, associated to the switched subsystem $\Sigma=(\mathbb X,\mathbb X_0,\mathbb X_s,P,\mathbb W, F,\mathbb Y,h)$, where $\mathbb X,\mathbb W$ are assumed to be finite unions of boxes. Let $\Sigma_{p}$ be $\delta$-ISS as in Definition \ref{def:SFD1}. Then one can construct a finite abstraction $T(\hat{\Sigma})=(\hat{X},\hat X_0,\hat X_s,\hat{U},\hat{W},\hat{\TF},\hat{Y},\hat{\op})$ where:
	\begin{itemize}
		\item $\hat{X}=\hat{\mathbb{X}}\times P\times \{0,\dots,k_d-1\}$, where $\hat{\mathbb{X}}=[\mathbb{X}]_{\eta}$ and $0 < \eta \leq \tup{min} \{span(\mathbb X_s)$, $span(\mathbb{X} \setminus \mathbb X_s)\}$ is the state set quantization parameter;
		\item  $\hat X_0=\hat{\mathbb X}_0\times  P \times \{0\}$, where $\hat{\mathbb X}_0 =[\mathbb{X}_0]_{\eta}$; 
		\item  $\hat X_s=\hat{\mathbb X}_s\times P\times  \{0,\cdots,k_d-1\}$, where $\hat{\mathbb X}_s = [\mathbb X_s]_{\eta}$; 
		\item $\hat{U}=U=P$;
		\item $(\hat{x}^+,p^+,l^+)\in \hat{\TF}((\hat{x},p,l),\hat{u},\hat{w})$ if and only if $\Vert f_p(\hat{x},\hat{w})-\hat{x}^+\Vert\leq \eta$, $\hat{u}=p$ and the following scenarios hold:
		\begin{itemize}
			\item $l<k_d-1$, $p^+=p$ and $l^+=l+1$;
			\item $l=k_d-1$, $p^+=p$ and $l^+=k_d-1$;
			\item $l=k_d-1$, $p^+\neq p$ and $l^+=0$;
		\end{itemize}
		\item   $\hat{Y} = \{{\op}(\hat{x},p,l)|(\hat{x},p,l)\in\hat{X}\}$;
		\item $\hat{\op}:\hat{X}\rightarrow \hat{Y}$, defined as $\hat{\op}(\hat{x},p,l)={\op}(\hat{x},p,l)=h(\hat{x})$;
		\item $\hat{W}=[\mathbb{W}]_{{\phi}}$,  where $\phi$, satisfying $0 \!<\!\Vert \phi \Vert\!\leq\! span(\mathbb W)$, is the internal input set quantization parameter.
	\end{itemize} 
\end{definition}
An illustration of the computation of the transitions of $T(\hat{\Sigma})$ is shown in Figure~\ref{abs}.

Note that in the case when the concrete switched subsystem $\Sigma$ admits a common $\delta$-ISS Lyapunov function as in Remark \ref{commonlya}, Definition \ref{smm} boils down to the following.
\begin{definition}\label{smm1}
	Consider a transition system $T(\Sigma)=(X,X_0,X_s,U,W,\TF,Y,\op)$, associated to the switched subsystem $\Sigma=(\mathbb X,\mathbb X_0,\mathbb X_s,P,\mathbb W, F,\mathbb Y,h)$, where $\mathbb X,\mathbb W$ are assumed to be finite unions of boxes. 
	Suppose $\Sigma$ admits a common $\delta$-ISS Lyapunov function as in Remark \ref{commonlya}. Then one can construct a finite abstraction $T(\hat{\Sigma})=(\hat{X},\hat X_0,\hat X_s,\hat{U},\hat{W},\hat{\TF},\hat{Y},\hat{\op})$ where:
	\begin{itemize}
		\item $\hat{X}=[\mathbb{X}]_{\eta}$, where $0 < \eta \leq \tup{min} \{span(\mathbb X_s)$, $span(\mathbb{X} \setminus \mathbb X_s)\}$ is the state set quantization parameter;
		\item  $\hat X_0=[\mathbb{X}_0]_{\eta}$; 
		\item  $\hat X_s = [\mathbb X_s]_{\eta}$; 
		\item $\hat{U} =P$;
		\item $\hat{x}^+\in \hat{\TF}(\hat{x},\hat{u},\hat{w})$ if and only if $\Vert f_{\hat{u}}(\hat{x},\hat{w})-\hat{x}^+\Vert\leq \eta$;
		\item $\hat{Y} = \{h(\hat{x})|\hat{x}\in\hat{X}\}$;
		\item $\hat{\op}:\hat{X}\rightarrow \hat{Y}$, defined as $\hat{\op}(\hat{x})=h(\hat{x})$;
		\item $\hat{W}=[\mathbb{W}]_{{\phi}}$,  where $\phi$, satisfying $0 \!<\!\Vert \phi \Vert\!\leq\! span(\mathbb W)$, is the internal input set quantization parameter.
	\end{itemize} 
\end{definition}

In order to construct a local $\varepsilon$-InitSOPSF from $T(\Sigma)$ to $T(\hat{\Sigma})$, we raise the following assumptions on functions $V_p$ appeared in Definition \ref{def:SFD1}, which are used to prove some of the main results later. 
\begin{assumption}\label{ass1} 
	There exists $\mu \geq 1$ such that
	\begin{align}\label{mue}
		\forall x,y \in \mathbb{X},~~ \forall p,q \in P,~~ V_p(x,y)\leq \mu  V_{q}(x,y). 
	\end{align}
\end{assumption}
Assumption \ref{ass1} is an incremental version of a similar assumption in \cite{VU} that is used to prove input-to-state stability of switched systems under constrained switching signals.
\begin{assumption}\label{ass2} 
	For all $p\in P$, there exists a $\mathcal{K}_{\infty}$ function $\gamma_p$ such that
	\begin{align}\label{tinq} 
		\forall x,y,z \in \mathbb{X},~~V_p(x,y)\leq V_p(x,z)+\gamma_p(\Vert y-z\Vert).
	\end{align}
\end{assumption}
Assumption \ref{ass2} is non-restrictive as shown in \cite{zamani2014symbolic} provided that one is interested to work on a compact subset of $\mathbb{X}$.   

Now, we establish the relation between $T(\Sigma)$ and $T(\hat{\Sigma})$, introduced above, via the notion of local $\varepsilon$-InitSOPSF as in Definition \ref{sf}.
\begin{theorem}\label{thm:2}
	Consider a dt-SS $\Sigma=(\mathbb X,\mathbb X_0,\mathbb X_s,P,\mathbb W,F,\mathbb Y,h)$  with its equivalent transition system  $T(\Sigma)=(X,X_{0},X_{S},U,W,\TF,Y,\mathcal H)$. Suppose $\Sigma_{p}$ is $\delta$-ISS as in Definition \ref{def:SFD1}, 
	with a function $V_p$ equipped with functions $\underline{\alpha}_p, \overline{\alpha}_p, \rho_{p}$ and constant $\kappa_p$,		
	and 
	Assumptions \ref{ass1} and \ref{ass2} hold. Let $\epsilon>1$. For any design parameters $\varepsilon, \vartheta \in \mathbb R_{\ge 0}$, let $T(\hat{\Sigma})$ be a finite abstraction of $T(\Sigma)$ constructed as in Definition \ref{smm}	with any quantization parameter $\eta$ satisfying 
	\begin{align} \label{secquantinit}	
		\eta \leq \min\{\hat{\gamma}^{-1}((1-\kappa)\varepsilon-\rho(\vartheta)),  \overline{\alpha}^{-1}(\varepsilon)\},
	\end{align}
	where $\kappa =\max\limits_{p\in P}\left\{\kappa^{\frac{\epsilon-1}{\epsilon}}_{p}\right\}$, $\rho=\max\limits_{p\in P}\left\{{\kappa^{-\frac{k_d}{\epsilon}}_p}{\rho_{p}}\right\}$, $\hat{\gamma}=\max\limits_{p\in P}\left\{{\kappa^{-\frac{k_d}{\epsilon}}_p}{\gamma_{p}}\right\}$, $\overline{\alpha}=\max\limits_{p\in P}\left\{{\kappa^{-\frac{l}{\epsilon}}_p}{\overline{\alpha}_p}\right\}$.
	If, $\forall p \in P, ~k_d\geq \epsilon \frac{\ln (\mu)}{\ln (\frac{1}{\kappa_p})}+1$, then function $\mathcal{V}$ defined as 
	\begin{align}\label{sm}
		\mathcal{V}((x,p,l),(\hat{x},p,l))\Let V_p(x,\hat{x})\kappa^{\frac{-l}{\epsilon}}_p,
	\end{align}
	is a local $\varepsilon$-InitSOPSF from $T(\Sigma)$ to $T(\hat{\Sigma})$.
\end{theorem}

\begin{proof} 
	We start by proving condition 1 in Definition \ref{sf}. Consider any initial and secret state $(x^0,p^0,0) \in {X}_0 \cap {X}_s$ in $T(\Sigma)$. From Definition \ref{smm}, for every $(x^0,p^0,0)\in {X}_0 \cap {X}_s$, there always exists $(\hat{x}^0,p^0,0) \in  \hat{X}_0 \cap \hat {X}_s$ such that $\Vert x^0-\hat x^0\Vert\leq\eta$.
	Hence, using \eqref{e:SFC11}, there exists $(\hat{x}^0,p^0,0) \in \hat{X}_0 \cap \hat{X}_s$ with $\mathcal{V}((x^0,p^0,0),(\hat{x}^0,p^0,0)) \leq \frac{\overline{\alpha}_p (\Vert x^0-\hat x^0\Vert )}{{\kappa^{\frac{l}{\epsilon}}_p}} \leq \frac{\overline{\alpha}_p (\eta) }{{\kappa^{\frac{l}{\epsilon}}_p}}$, and condition 1(a) is satisfied with $\overline{\alpha}=\max\limits_{p\in P}\left\{{\kappa^{-\frac{l}{\epsilon}}_p}{\overline{\alpha}_p}\right\}$ and $\overline{\alpha}(\eta) \leq \varepsilon$ by \eqref{secquantinit}. For every $(\hat{x}^0,p^0,0) \in \hat {X}_0 \setminus \hat {X}_s$, by choosing $x^0 = \hat x^0$ with $(x^0,p^0,0)$ also being inside ${X}_0 \setminus {X}_s$, we get $\mathcal{V}((x^0,p^0,0),(\hat{x}^0,p^0,0)) = 0 \leq \varepsilon$. Hence, condition 1(b) in Definition \ref{sf} holds as well.
	
	Next, we show condition 2 in Definition \ref{sf} holds.	
	Given the Lipschitz assumption on $h$ and since, $\forall p\in P$, $\Sigma_p$ is $\delta$-ISS, from \eqref{e:SFC11}, $\forall (x,p,l)\in X$ and $ \forall (\hat{x},p,l) \in \hat{X}
	$, we have 
	\begin{align*}
		&\Vert \op(x,p,l)-\hat{\op}(\hat{x},p,l)\Vert=\Vert h(x)-\hat{h}(\hat{x})\Vert
		\leq \ell(\Vert x-\hat{x}\Vert)\leq\ell\circ\underline{\alpha}^{-1}_p(V_p(x,\hat{x}))
		=\ell\circ\underline{\alpha}^{-1}_p\left({\kappa^{\frac{l}{\epsilon}}_p}\mathcal{V}((x,p,l),(\hat{x},p,l))\right)\\
		&\leq\ell\circ\underline{\alpha}^{-1}_p\left(\mathcal{V}((x,p,l),(\hat{x},p,l))\right)
		\leq\hat{\alpha}\left(\mathcal{V}((x,p,l),(\hat{x},p,l))\right),
	\end{align*}	
	where $\hat{\alpha}=\max\limits_{p\in P}\{\ell\circ\underline{\alpha}^{-1}_p\}$.  
	By defining $\alpha=\hat{\alpha}^{-1}$, one obtains
	\begin{align}\notag
		\alpha (\Vert \op(x,p,l)-\hat{\op}(\hat{x},p,l)\Vert )\leq \mathcal{V}((x,p,l),(\hat{x},p,l)),
	\end{align}satisfying condition 2. 
	Now we show condition 3 in Definition \ref{sf}. From \eqref{tinq}, $\forall x\in \mathbb{X}, \forall \hat{x} \in \mathbb{\hat{X}}$, $\forall w \in {W},\forall \hat{w} \in {\hat{{W}}}$, we have 
	\begin{align*}
		V_p(f_p(x,w),\hat{x}^+)\leq V_p(f_p(x,w),f_p(\hat{x},\hat{w}))+\gamma_p(\Vert \hat{x}^+-f_p(\hat{x},\hat{w})\Vert),
	\end{align*}
	for any $\hat{x}^+$ such that $(\hat{x}^+,p^+,l^+)\in \hat{\TF}((\hat{x},p,l),\hat{u},\hat{w})$.
	Now, from Definition \ref{smm}, the above inequality reduces to
	\begin{align*}
		V_p(f_p(x,w),\hat{x}^+)\leq V_p(f_p(x,w),f_p(\hat{x},\hat{w}))+\gamma_p(\eta).
	\end{align*}
	Note that by \eqref{e:SFC22}, one gets 
	\begin{align*}
		V_p(f_p(x,w),f_p(\hat x,\hat w))
		\leq \kappa_p V_p(x,\hat{x})+\rho_{p}(\Vert w- \hat{w}\Vert ).
	\end{align*}	
	Hence, $\forall x\in \mathbb{X}, \forall \hat{x} \in \mathbb{\hat{X}}$, $\forall w \in {W},\forall \hat{w} \in {\hat{{W}}}$, one obtains
	\begin{align}\label{fe}
		V_p(f_p(x,w),\hat{x}^+)
		\leq \kappa_p V_p(x,\hat{x})+\rho_{p}(\Vert w- \hat{w}\Vert )+\gamma_p(\eta),
	\end{align}
	for any $\hat{x}^+$ such that $(\hat{x}^+,p^+,l^+)\in \hat{\TF}((\hat{x},p,l),\hat{u},\hat{w})$.
	Now, in order to show function $\mathcal{V}$ defined in \eqref{sm} satisfies condition 3 in Definition \ref{sf}, we consider the different scenarios in Definition \ref{smm}: 
	\begin{itemize}
		\item$l<k_d-1$, $p^+=p$ and $l^+=l+1$, using \eqref{fe} and $k_d>l+1$, we have
		\begin{small}
			\begin{align*}
				\mathcal{V}&((x^+,p^+,l^+),(\hat{x}^+,p^+,l^+))=\frac{V_{p^+}(x^+,\hat{x}^+)}{\kappa^{\frac{l^+}{\epsilon}}_p}=\frac{V_p(f_p(x,w),\hat{x}^+)}{\kappa^{\frac{l+1}{\epsilon}}_p}\\
				&\leq\frac{\kappa_pV_p(x,\hat{x})+\rho_{p}(\Vert w- \hat{w}\Vert )+\gamma_p(\eta)}{\kappa^{\frac{l+1}{\epsilon}}_p}=\frac{\kappa_p}{\kappa^{\frac{1}{\epsilon}}_p}\frac{V_p(x,\hat{x})}{\kappa^{\frac{l}{\epsilon}}_p}+\frac{\rho_{p}(\Vert w- \hat{w}\Vert )+\gamma_p(\eta)}{\kappa^{\frac{l+1}{\epsilon}}_p}\\
				&\leq\kappa^{\frac{\epsilon-1}{\epsilon}}_p\mathcal{V}((x,p,l),(\hat{x},p,l))+\frac{\rho_{p}(\Vert w- \hat{w}\Vert) +\gamma_p(\eta)}{\kappa^{\frac{k_d}{\epsilon}}_p}.
			\end{align*}
		\end{small}
		\item $l=k_d-1$, $p^+=p$ and $l^+=k_d-1$, using \eqref{fe} and $\frac{\epsilon-1}{\epsilon}<1$, one gets
		\begin{small}
			\begin{align*}
				\mathcal{V}&((x^+,p^+,l^+),(\hat{x}^+,p^+,l^+))=\frac{V_{p^+}(x^+,\hat{x}^+)}{\kappa^{\frac{l^+}{\epsilon}}_p}=\frac{V_p(f_p(x,w),\hat{x}^+)}{\kappa^{\frac{l}{\epsilon}}_p}\\
				&\leq\frac{\kappa_pV_p(x,\hat{x})+\rho_{p}(\Vert w- \hat{w}\Vert )+\gamma_p(\eta)}{\kappa^{\frac{l}{\epsilon}}_p}=\kappa_p\frac{V_p(x,\hat{x})}{\kappa^{\frac{l}{\epsilon}}_p}+\frac{\rho_{p}(\Vert w- \hat{w}\Vert) +\gamma_p(\eta)}{\kappa^{\frac{l}{\epsilon}}_p}\\
				&\leq\kappa^{\frac{\epsilon-1}{\epsilon}}_p\mathcal{V}((x,p,l),(\hat{x},p,l))+\frac{\rho_{p}(\Vert w- \hat{w}\Vert) +\gamma_p(\eta)}{\kappa^{\frac{k_d}{\epsilon}}_p}.
			\end{align*}
		\end{small}
		\item $l=k_d-1$, $p^+\neq p$ and $l^+=0$, using \eqref{fe}, $\mu{\kappa^{\frac{k_d-1}{\epsilon}}_{p}}\leq1$, and $\frac{\epsilon-1}{\epsilon}<1$, one has
		\begin{small}
			\begin{align*}
				\mathcal{V}&((x^+,p^+,l^+),(\hat{x}^+,p^+,l^+))=\frac{V_{p^+}(x^+,\hat{x}^+)}{\kappa^{\frac{l^+}{\epsilon}}_{p^+}}\leq\mu V_{p}(f_{p}(x,w),\hat{x}^+)\\
				&\leq\frac{\mu\kappa^{\frac{k_d-1}{\epsilon}}_p\left(\kappa_{p}V_{p}(x,\hat{x})+\rho_{p}(\Vert w- \hat{w}\Vert )+\gamma_{p}(\eta)\right)}{\kappa^{\frac{k_d-1}{\epsilon}}_p}=\kappa_p\frac{V_p(x,\hat{x})}{\kappa^{\frac{l}{\epsilon}}_p}+\frac{\rho_{p}(\Vert w- \hat{w}\Vert) +\gamma_p(\eta)}{\kappa^{\frac{l}{\epsilon}}_p}\\
				&\leq\kappa^{\frac{\epsilon-1}{\epsilon}}_p\mathcal{V}((x,p,l),(\hat{x},p,l))+\frac{\rho_{p}(\Vert w- \hat{w}\Vert) +\gamma_p(\eta)}{\kappa^{\frac{k_d}{\epsilon}}_p}.
			\end{align*}
		\end{small}
	\end{itemize}
	Note that $\forall p\in P,\mu{\kappa^{\frac{k_d-1}{\epsilon}}_{p}}\leq1$, since $\forall p \in P, k_d\geq \epsilon \frac{\ln (\mu)}{\ln (\frac{1}{\kappa_p})}+1$.
	Hence, $\forall (x,p,l)\in X$, $\forall (\hat{x},p,l) \in \hat{X}$, $\forall w \in W$, $\forall \hat{w} \in {\hat{W}}$, one gets
	\begin{align} \label{fe1}
		\mathcal{V}((x^+,p^+,l^+),(\hat{x}^+,p^+,l^+))\leq\kappa\mathcal{V}((x,p,l),(\hat{x},p,l))+\rho(\Vert w- \hat{w}\Vert)+\hat{\gamma}(\eta).
	\end{align}Now, we show the condition 3(a) in Definition \ref{sf} holds. Let us consider any pair of states $(x,p,l)\in X$, $(\hat{x},p,l) \in \hat{X}$, satisfying $\mathcal{V}((x,p,l),(\hat{x},p,l)) \leq \varepsilon$, and any $w\in W$, $\hat w \in \hat W$ such that $\Vert {w}-\hat{w}\Vert \leq \vartheta$. Combining \eqref{fe1} with \eqref{secquantinit} for any $(x^+,p^+,l^+)\in \TF((x,p,l),u,w)$ and any 
	$(\hat{x}^+,p^+,l^+)\in \hat{\TF}((\hat{x},p,l),\hat{u},\hat{w})$ with $\hat{u} = u$, one obtains: 
	\begin{align}\label{sos}
		\mathcal{V}((x^+,p^+,l^+),(\hat{x}^+,p^+,l^+)) \leq \kappa\varepsilon+\rho(\vartheta)+\hat{\gamma}(\hat{\gamma}^{-1}((1-\kappa)\varepsilon-\rho(\vartheta))) = \varepsilon,
	\end{align}which shows that condition 3(a) is satisfied. Similarly, for any  $(\hat{x}^+,p^+,l^+)\in \hat{\TF}((\hat{x},p,l),\hat{u},\hat{w})$, condition 3(b)  is also satisfied using the same reasoning with $(x^+,p^+,l^+)\in \TF((x,p,l),\hat{u},w)$. 
	Therefore, we conclude that $\mathcal{V}$ is a local $\varepsilon$-InitSOPSF from $T(\Sigma)$ to $T(\hat{\Sigma})$ .		
\end{proof}	
\begin{remark}\label{nonl-common}
	If $\Sigma$ admits a common $\delta$-ISS Lyapunov function satisfying Assumption \ref{ass2}, then function $\mathcal{V}$ in Theorem \ref{thm:2} reduces to $\mathcal{V}((x,p,l),(\hat{x},p,l))\Let V(x,\hat{x})$.
\end{remark}

Given the results of Theorems \ref{thm:3} and \ref{thm:2}, one can see that conditions \eqref{compoquaninit} and \eqref{secquantinit} may not hold at the same time. In the following subsection, we will discuss about the inherent property that the network should have such that one can design suitable quantization parameters to satisfy conditions \eqref{compoquaninit} and \eqref{secquantinit} simultaneously.

\subsection{Compositionality Result}
%
We raise the following assumption which provides a small-gain type condition so that one can verify whether the competing conditions \eqref{compoquaninit} and \eqref{secquantinit} can be satisfied simultaneously.

\begin{assumption}\label{assump}
	Consider network $\mathcal{I}(T(\Sigma_1),\ldots,T(\Sigma_N))$ induced by $N\in\N_{\ge1}$ transition systems~$T(\Sigma_i)$. Assume that each $T(\Sigma_i)$ and its finite abstraction $T(\hat \Sigma_i)$ admit a local $\varepsilon_i$-InitSOPSF $\mathcal{V}_i$ defined in \eqref{sm}, associated with functions and constants $\kappa_i$, $\alpha_i$, and $\rho_{i}$ that appeared in Theorem \ref{thm:2}. Define
	\begin{align}\label{gammad}
		\gamma_{ij}&\Let\left\{
		\begin{array}{lr}
			(1-\kappa_{i})^{-1}\rho_i\circ\alpha_{j}^{-1} &\mbox{if } j \in \mathcal{N}_i, \\
			0 &\mbox{otherwise},
		\end{array}\right.
	\end{align}for all $ i,j \in [1;N]$,
	and assume that functions $\gamma_{ij}$ defined in \eqref{gammad} satisfy
	\begin{align}\label{SGC}
		\gamma_{i_1i_2}\circ\gamma_{i_2i_3}\circ\cdots\circ\gamma_{i_{r-1}i_r}\circ\gamma_{i_ri_1}<\mathcal{I}_d,
	\end{align}
	$\forall(i_1,\ldots,i_r)\in\{1,\ldots,N\}^r$, where $r\in \{1,\ldots,N\}$.
\end{assumption}

Now, we show that, under the above small-gain assumption, one can always compositionally design local quantization parameters to satisfy conditions \eqref{compoquaninit} and \eqref{secquantinit} simultaneously.

\begin{theorem} \label{smallgain}
	Suppose that Assumption \ref{assump} holds. Then, 
	there always exist local quantization parameters $\eta_i$ and $\phi_{ij}$, $\forall i,j \in [1;N]$, as designed in Algorithm \ref{quantialgo}, such that \eqref{compoquaninit} and \eqref{secquantinit} can be satisfied simultaneously.
\end{theorem}
\begin{proof}
	First, let us note that the small-gain condition \eqref{SGC} implies that $\exists \sigma_i \in \mathcal{K}_{\infty}$ satisfying $\forall i \in [1;N]$,  
	\begin{align}\label{gamcur}
		&\max\limits_{j \in \mathcal{N}_i}\{\gamma_{ij}\circ\sigma_j\}<\sigma_i, 
	\end{align}see \cite[Theorem 5.2]{090746483}. Then, from \eqref{gammad}, we have $\forall i\in [1;N]$,
	\begin{align} \label{inequalityinscc}
		&\max\limits_{j \in \mathcal{N}_i}\{\gamma_{ij}\circ\sigma_j\}<\sigma_i\Longrightarrow \max\limits_{j \in \mathcal{N}_i}\{(1-\kappa_{i})^{-1}\rho_{i}\circ\alpha_{j}^{-1}\circ\sigma_j\}<\sigma_i  \Longrightarrow \rho_{i} \circ\max\limits_{j \in \mathcal{N}_i}\{\alpha_{j}^{-1}\circ\sigma_j\}< (1-\kappa_{i})\sigma_i.
	\end{align}Next, suppose that we are given a sequence of functions $\sigma_i \in \mathcal{K}_{\infty}$, $\forall i\in [1;N]$, satisfying \eqref{gamcur}. 
	Assume we are given any desired precision $\varepsilon$ as in Definition \ref{sfg}.	Let us set $\varepsilon_i=\sigma_i(r)$, $\forall i \in [1;N]$, where $r \in \mathbb{R}_{>0}$ is chosen such that $\max\limits_{i}\{\sigma_i(r)\} = \varepsilon$. Then, we choose internal input quantization paramters $\phi_{ij}$, $\forall i,j \in [1;N]$, such that 
	\begin{align} \label{interinputquan}
		\max\limits_{j \in \mathcal{N}_i}\{\phi_{ij}\} < \rho_{i}^{-1}((1-\kappa_{i})\varepsilon_i) - \max\limits_{j \in \mathcal{N}_i}\{\alpha_{j}^{-1}(\varepsilon_j)\}.
	\end{align}Now, by setting $\vartheta_i = \max\limits_{j \in \mathcal{N}_i}\{\alpha^{-1}_j(\varepsilon_j)+\phi_{ij}\}$, and combining \eqref{inequalityinscc} and \eqref{interinputquan}, one has $\forall i \in [1;N]$ 
	\begin{align}  \label{posiinequal}
		\rho_{i}(\vartheta_i) &	= \rho_{i}(\max\limits_{j \in \mathcal{N}_i}\{\alpha^{-1}_j(\varepsilon_j)+\phi_{ij}\}) \leq \rho_{i}(\max\limits_{j \in \mathcal{N}_i}\{\alpha^{-1}_j(\varepsilon_j)+\max\limits_{j \in \mathcal{N}_i}\{\phi_{ij}\})< (1-\kappa_{i})\varepsilon_i.
	\end{align}Thus, by \eqref{posiinequal}, given any pair of parameters $(\varepsilon_i, \vartheta_i)$, one can always find suitable local parameters $\eta_i$ to satisfy \eqref{secquantinit}. Additionally, the selection of $\vartheta_i = \max\limits_{j \in \mathcal{N}_i}\{\alpha^{-1}_j(\varepsilon_j)+\phi_{ij}\}$ ensures that \eqref{compoquaninit} is satisfied as well, which concludes the proof.
\end{proof}

\begin{algorithm}[h!]
	\DontPrintSemicolon
	
	\KwInput{The desired precision $\varepsilon \in \mathbb{R}_{>0}$; the simulation functions $\mathcal{V}_i$ equipped with functions $\kappa_i$, $\alpha_i$, $\rho_{i}$, $\hat{\gamma}_i$, and $\overline{\alpha}_i$, $\forall i \in [1;N]$; functions $\sigma_{_i}$, $\forall i \in [1;N]$, satisfying \eqref{gamcur}.}

	Choose $r \in \mathbb{R}_{>0}$ s.t. $\max\limits_{i\in [1;N]}\{\sigma_i(r)\} = \varepsilon$;
	
	Set $\varepsilon_i=\sigma_i(r)$, $\forall i \in [1;N]$;
	
	Design $\phi_{ij}\in \mathbb{R}_{>0}$ s.t. $\max\limits_{j \in \mathcal{N}_i}\{\phi_{ij}\} < \rho_{i}^{-1}((1-\kappa_{i})\varepsilon_i) - \max\limits_{j \in \mathcal{N}_i}\{\alpha_{j}^{-1}(\varepsilon_j)\},  \forall i,j  \in  [1; N]$; 
	
	Set $\vartheta_i = \max\limits_{j \in \mathcal{N}_i}\{\alpha^{-1}_j(\varepsilon_j)+\phi_{ij}\}$, $\forall i \in [1;N]$;
	
	Design $\eta_i\!\in\! \mathbb{R}_{>0}$ s.t. 
	$\eta_i \leq \min\{\hat{\gamma}_i^{-1}((1-\kappa_i)\varepsilon_i-\rho_i(\vartheta_i)),  \overline{\alpha}_i^{-1}(\varepsilon_i)\}$;  
	
	\KwOutput{Quantization parameters $\eta_i \in \mathbb{R}_{>0}$ and $\phi_{ij} \in \mathbb{R}_{>0}$, $\forall i \in [1;N]$.}					
	
	\caption{Compositional design of local quantization parameters $\eta_i \!\in\! \mathbb{R}_{>0}$ and $\phi_{ij} \!\in\! \mathbb{R}_{>0}$, $\!\forall i \!\in\!\! [1;N]$.} \label{quantialgo}
\end{algorithm}

\begin{figure}[ht!]
	\centering
	\includegraphics[width=0.8\textwidth]{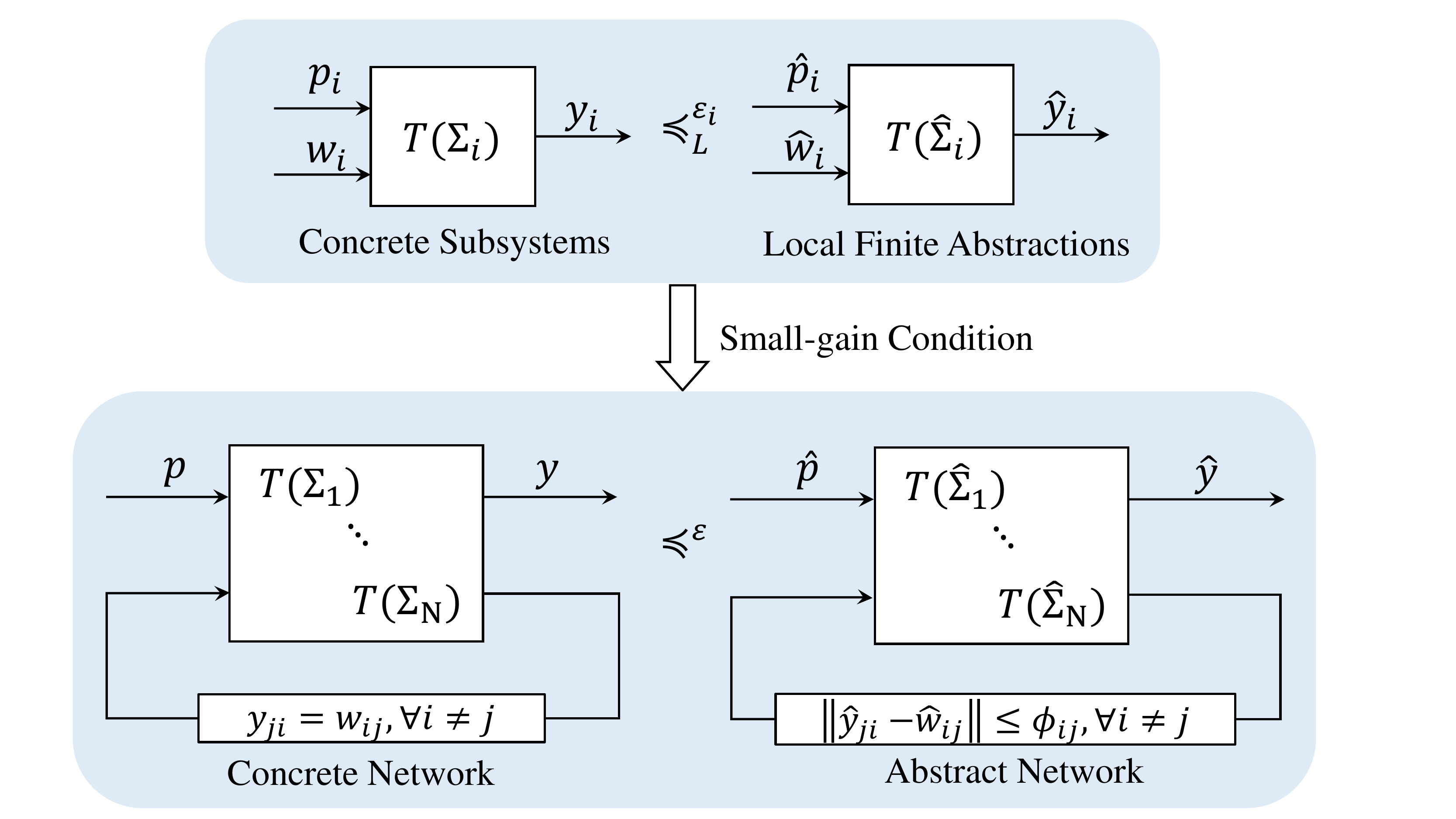}
	\caption{Compositionality result.}
	\label{compsitionality}
\end{figure}
\begin{remark}
	The compositionality result in Theorem \ref{smallgain} imposes a small-gain type condition on the concrete network of switched subsystems for the existence of proper compositional finite abstraction, as depicted in Figure~\ref{compsitionality}. In particular, under such small-gain type conditions, one can always find suitable local quantization parameters to construct local finite abstractions. The interconnection of the local finite abstractions can be used to serve as a finite abstraction for the concrete network satisfying the simulation relation ${T}(\Sigma) \preceq^{\varepsilon} T(\hat{\Sigma})$.
\end{remark}
\begin{remark}
	Let us provide a general guideline on the computation of $\mathcal{K}_{\infty}$ functions $\sigma_i$, $i\in[1;N]$, that are used in Theorem \ref{smallgain}: 
	$(i)$ in a general case when the network is consisting of $N\ge 1$ subsystems, functions $\sigma_i, i\in[1;N]$, can be constructed numerically by leveraging the algorithm introduced in \cite{Eaves} and the technique presented in \cite[Proposition 8.8]{090746483}, see \cite[Chapter 4]{Rufferp}; 
	$(ii)$ for the case of having two and three subsystems in the network, there have been some construction techniques proposed in \cite{JIANG} and  \cite[Section 9]{090746483}, respectively;  
	$(iii)$ when the gain functions appeared in \eqref{gammad} satisfy $\gamma_{ij}<\mathcal{I}_d$, $\forall i,j\in [1;N]$, then one can always choose $\sigma_i, i\in[1;N]$ to be identity functions.
\end{remark}

\section{Illustrative Example} \label{sec:example}

Here, we provide an illustrative example to show how one can leverage the proposed compositional approach to check approximate initial-state opacity of a network of switched systems based on its finite abstraction.

Consider a network of discrete-time switched systems $\Sigma=(\mathbb X,\mathbb X_0,\mathbb X_s,P,F,$ $\mathbb Y,h)$ as in Definition \ref{netsw}, consisting of $n$ subsystems $\Sigma_i$ each described by:
\begin{align}\label{exsm}
	\Sigma_i:\left\{
	\begin{array}[\relax]{rl}
		\mathbf{x}_i(k+1)&=a_{i{\mathsf p}_i(k)}\mathbf{x}_i(k)+ d_i\omega_i(k)+b_{i\mathsf{p}_i(k)},\\
		\mathbf{y}_i(k)&= c_i \mathbf{x}_i(k),
	\end{array}\right.
\end{align}
where $\mathsf p_i (k)\in P_i = \{1,2\}$, $\forall k \in \N$,  denotes the modes of each subsystem $\Sigma_i$. 
The other parameters are as the following: $a_{i1} = 0.05$, $a_{i2} = 0.1$, $b_{i1}=0.1$, $b_{i2}=0.15$, $d_i=0.05$, $c_i = [c_{i1};\dots;c_{in}]$ with $c_{i(i+1)} = 1$, $c_{ij} = 0$, $\forall i\in [1;n-1], \forall j \neq i+1$, $c_{n1}=c_{nn} = 1$, $c_{nj}= 0$, $\forall j \in [2;n-1]$. The internal inputs are subject to the constraints $\omega_1(k) = c_{n1}\mathbf{x}_n(k)$ and $\omega_i(k) = c_{(i-1)i}\mathbf{x}_{(i-1)}(k)$, $\forall i\in [2;n]$.
For each switched subsystem, the state set is $\mathbb X_i =\mathbb X_{0_i} = (0, 0.6)$, $\forall i \in [1;n]$, the secret set is $\mathbb{X}_{s_1} = (0,0.2]$, $\mathbb{X}_{s_2} = [0.4,0.6)$, $\mathbb{X}_{s_i} = (0, 0.6)$, $\forall i \in [3;n]$, 
the output set is $\mathbb Y_i=\prod_{j=1}^n  \mathbb Y_{ij}$ where
$\mathbb Y_{i(i+1)} =(0, 0.6)$, $\mathbb Y_{ii} =\mathbb Y_{ij} =\{0\}$, $\forall i \in [1;n-1]$, $\forall j \neq i+1$, $\mathbb Y_{nn} =\mathbb Y_{n1} = (0, 0.6)$, $\mathbb Y_{nj} =\{0\}$, $\forall j \in [2;n-1]$, and internal input set is $\mathbb{W}_1=\mathbb{Y}_{ni}$, $\mathbb{W}_i=\mathbb{Y}_{(i-1)i}$, $\forall i\in [2;n]$.
Intuitively, the output of the network is the external output of the last subsystem $\Sigma_n$. The interconnection topology of the network is depicted in Figure~\ref{net}.
\begin{figure}[h!]
	\centering 
	\includegraphics[scale=0.45]{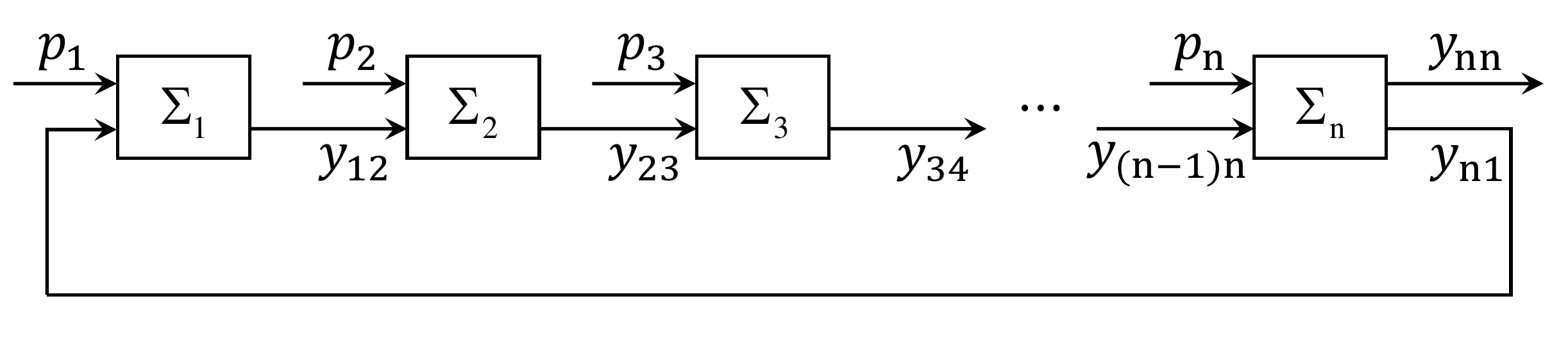}
	\caption{The interconnection topology of the network of discrete-time switched subsystems $\Sigma_i$.}
	\label{net}
\end{figure}

The main goal of this example is to check approximate initial-state opacity of the concrete network using its finite abstraction. 
Now, let us construct a finite abstraction of $\Sigma$ compositionally with accuracy $\hat \varepsilon = 0.25$ as defined in \eqref{er}, which preserves approximate initial-state opacity. We implement our compositional approach to achieve this goal.

Consider functions $V_{ip_i} = |x_i- \hat x_i|$, $\forall i \in [1;n]$. It can be readily verified that \eqref{e:SFC11} and \eqref{e:SFC22} are satisfied with 
$\underline{\alpha}_{ip_i}=\overline{\alpha}_{ip_i}=\mathcal{I}_d$, $\rho_{ip_i}=0.05$, $\forall p_i\in P_i$,  $\kappa_{i1}= a_{i1}=0.05$,  $\kappa_{i2}= a_{i2}=0.1$.
Condition \eqref{tinq} is satisfied with $\gamma_{ip_i}=\mathcal I_d$, $\forall p_i\in P_i$. Moreover, since $V_{ip_i}= V_{ip^+_i},\forall p_i,p^+_i\in P_i$, $V_i(x_i,\hat{x}_i)= |x_i-\hat x_i|$ is a common $\delta$-ISS Lyapunov function for subsystem $\Sigma_i$. 
Next, given functions $\kappa_i = 0.1$, $\rho_i=0.06\mathcal I_d$, $\alpha_i= \mathcal{I}_d$, $\hat{\gamma}_i=1.05\mathcal I_d$, $\overline{\alpha}_i=\mathcal{I}_d$ as appeared in Theorem \ref{thm:2}, we have $\gamma_{ij} < \mathcal{I}_d$ by \eqref{gammad}, $\forall i,j \in [1;n]$.
Hence, the small-gain condition \eqref{SGC} is satisfied. Then, by applying Theorem \ref{smallgain} and choosing functions $\sigma_i = \mathcal{I}_d$, $\forall i \in [1;n]$, such that \eqref{gamcur} holds, we obtain proper pairs of local parameters $(\varepsilon_i,\vartheta_i)= (0.25, 0.25)$ for all of the transition systems. 
Accordingly, we provide a suitable choice of local quantization parameters as $\eta_i=0.2$, $\forall i \in [1;n]$, such that inequality \eqref{secquantinit} for each transition system $T( {\Sigma}_i)$ is satisfied. 
Then, we construct local finite abstractions $T(\hat{\Sigma}_i)=(\hat{X}_i,\hat X_{0_i},\hat X_{s_i},\hat{U}_i,\hat{W}_i,\hat{\TF}_i,\hat{Y}_i,\hat{\op}_i)$ as in Definition \ref{smm1}, where:
\begin{align*}
	&\hat{X}_i= \hat{{X}}_{0_i} = \{0.2,0.4\}, \forall i \in [1;n],\\
	&\hat{X}_{s_i} = \left\{ 
	\begin{array}[\relax]{ll} 
		\{0.2\}, \quad  \quad \quad  \quad \quad  \quad&\mbox{if } i = 1\\
		\{0.4\}, \quad  \quad  \quad  \quad \quad  \quad&\mbox{if } i = 2\\
		\{0.2,0.4\}, \quad  \quad \quad  \quad \quad  \quad&\mbox{otherwise}
	\end{array}\right.\\
	&\hat Y_{i} =  \left\{ 
	\begin{array}[\relax]{ll} 
		\prod_{j=1}^{i} \{0\}\!\times\!\{0.2,0.4\}\!\times\!\!\prod_{j=i+2}^{n}\{0\}, \quad  \quad &\mbox{if } i \in [1;n\!-\!1]\\
		\{0.2,0.4\} \!\times\! \prod_{j=2}^{n-1}\{0\} \!\times\!\{0.2,0.4\}, \quad  \quad  &\mbox{otherwise}
	\end{array}\right.\\
	&\hat W_i=\{0.2,0.4\}, \forall i \in [1;n].
\end{align*}
Using the result in Theorem \ref{thm:2}, one can verify that $V_i(x_i,\hat{x}_i)= |x_i-\hat x_i|$ is a local $\varepsilon_i$-InitSOPSF from each $T(\Sigma_i)$ to its finite abstraction $T(\hat{\Sigma}_i)$. Furthermore, by the compositionality result in Theorem \ref{thm:3}, we obtain that ${V} = \max\limits_{i}\{V_i(x_i,\hat{x}_i) \} = \max\limits_{i}\{|x_i- \hat x_i|\}$ is an $\varepsilon$-InitSOPSF from $T(\Sigma)=\mathcal{I}(T(\Sigma_1),\ldots,$ $T(\Sigma_{n}))$ to $T(\hat{\Sigma})=\hat{\mathcal{I}}(T(\hat{\Sigma}_1),\ldots,T(\hat{\Sigma}_{n}))$ with  $\varepsilon = \max\limits_{i} \varepsilon_i = 0.25$.

\begin{figure}[ht!]
	\centering
	\begin{small}
		\begin{tikzpicture}[->,>=stealth',shorten >=1pt,auto,node distance=2.1cm, inner sep=1pt, initial text =,
		every state/.style={draw=black,fill=white,state/.style=state with output},
		accepting/.style={draw=black,thick,fill=red!80!green,text=white},bend angle=25]
		
		\node[] at (3.5,0.3) {$T(\hat \Sigma_1)$:};
		\node[] at (3.5,-1.8) {$T(\hat \Sigma_2)$:};
		\node[] at (3.5,-3.9) {$T(\hat \Sigma_3)$:};

		\node[state with output,initial] (A0)  at (5.5,0.3)          {$q_1$ \nodepart{lower} $0y0$};
		\node[state with output,accepting,initial right]         (B0) [right of=A0] {$q_2$ \nodepart{lower} $0Y0$};
		
		\node[state with output,accepting,initial] (C0)   [below of=A0]                 {$q_1$ \nodepart{lower} $00y$};
		\node[state with output,initial right]         (D0) [right of=C0] {$q_2$ \nodepart{lower} $00Y$};

		\node[state with output,initial] (E0)   [below of=C0]                 {$q_1$ \nodepart{lower} $y0y$};
		\node[state with output,initial right]         (F0) [right of=E0] {$q_2$ \nodepart{lower} $Y0Y$};

		\path (A0) edge [loop above]    node {(2,Y)} (A0)
		edge                    node {} (B0)
		(B0) edge      [loop above]    node {} (B0)
		
		(C0) edge [loop above]    node {(2,Y)} (C0)
		edge           node {} (D0)
		(D0) edge     [loop above]  node {} (D0)

		(E0) edge [loop above]    node {(2,Y)} (E0)
		edge           node {} (F0)
		(F0) edge     [loop above]  node {} (F0)
		
		;
		\end{tikzpicture}
		\caption{Local finite abstractions of transition systems.}
		\label{exautomata1}
	\end{small}
	\vspace{-0.5cm}
\end{figure}
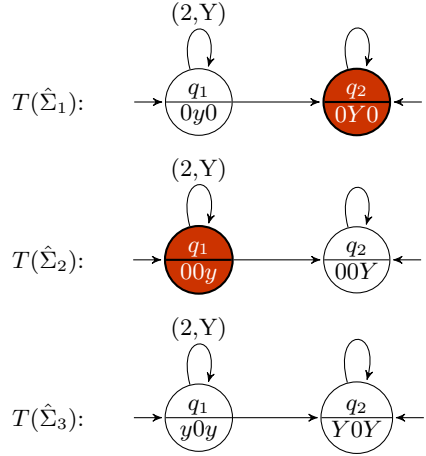

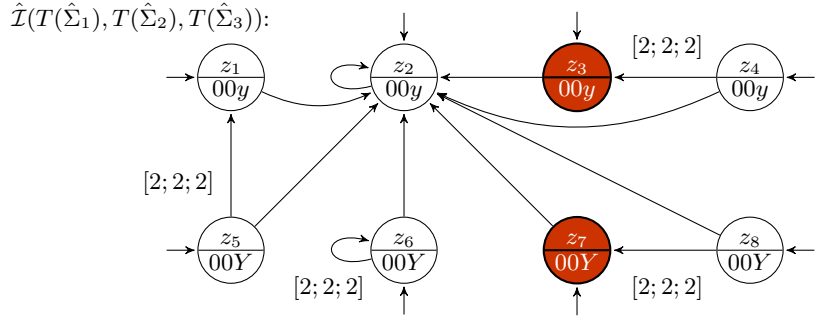
\begin{figure}[ht!]
	\centering
	\begin{small}
		
		\begin{tikzpicture}[->,>=stealth',shorten >=1pt,auto,node distance=2.3cm, inner sep=1pt, initial text =,
		every state/.style={draw=black,fill=white,state/.style=state with output},
		accepting/.style={draw=black,thick,fill=red!80!green,text=white},bend angle=25]
		\node[] at (2,0.8) {$\hat {\mathcal{I}}(T(\hat \Sigma_1),T(\hat \Sigma_2),T(\hat \Sigma_3))$:};
		\node[] at (9.0,0.4) {$[2;2;2]$};
		\node[] at (2.5,-1.4) {$[2;2;2]$};
		\node[] at (4.5,-2.8) {$[2;2;2]$};
		\node[] at (9.0,-2.8) {$[2;2;2]$};

		\node[state with output,initial above] (A)      at (5.5,0.0)              {$z_2$ \nodepart{lower} $00y$};
		\node[state with output, accepting,initial above]         (B) [right of=A] {$z_3$ \nodepart{lower} $00y$};
		\node[state with output,initial]         (D) [left of=A] {$z_1$ \nodepart{lower} $00y$};
		\node[state with output,initial right]         (E) [right of=B] {$z_4$ \nodepart{lower} $00y$};
		\node[state with output,initial below]         (H) [below of=A]       {$z_6$ \nodepart{lower} $00Y$};
		\node[state with output,initial]         (F) [left of=H] {$z_5$ \nodepart{lower} $00Y$};
		\node[state with output, accepting,initial below]         (G) [right of=H] {$z_7$ \nodepart{lower} $00Y$};
		\node[state with output,initial right]         (I) [right of=G]       {$z_8$ \nodepart{lower} $00Y$};

		\path (A) edge  [loop left]    node {} (A)
		(B) edge              node {} (A)
		(D) edge      [bend right]        node {} (A)
		(E) edge          node {} (B)
		edge    [bend left]      node {} (A)
		(F) edge          node {} (D)
		edge           node {} (A)
		
		(G) edge          node {} (A)
		(H) edge [loop left] node {} (H)
		edge    node {} (A)
		(I)  edge          node {} (G)
		edge    node {} (A);
		\end{tikzpicture}
	\end{small}
	\caption{Finite abstraction of a network of 3 transition systems.}
	\label{exautomata2}
	
\end{figure}

Now, let us verify approximate initial-state opacity for $T(\Sigma)$ using the network of finite abstractions $T(\hat{\Sigma})$. To do this, we first show an example of a network consisting of $3$ transition systems, as shown in Figures~\ref{exautomata1} and \ref{exautomata2}. The three automata in Figure~\ref{exautomata1}  represent the finite abstractions of the local transition systems, and the one in Figure~\ref{exautomata2} is the network of finite abstractions. Each circle is labeled by the state (top half) and the corresponding output (bottom half). Initial states are distinguished by being the target of a sourceless arrow. The states marked in red represent the secret states. The symbols on the edges show the switching signals $\mathsf{p}(k) \in \{1,2\}^{3}$ and internal inputs coming from other local transition systems. For simplicity of demonstration, we use symbols to represent the state and output vectors, where the states of local transition systems are denoted by $q_1=[0.4]$, $q_2=[0.2]$, the states of network of transition systems are denoted by
\begin{align*}
	z_1=[q_1;q_2;q_2], z_2=[q_2;q_2;q_2], z_3=[q_2;q_1;q_2], z_4=[q_1;q_1;q_2], \\
	z_5=[q_2;q_2;q_1], z_6=[q_2;q_1;q_1], z_7=[q_1;q_1;q_1], z_8=[q_1;q_2;q_1], 
\end{align*}
and the outputs of the corresponding states are represented as $y=0.2$ and $Y=0.4$ with the symbols like $00y = [0;0;0.2]$, $00Y= [0;0;0.4]$ representing concatenated output vectors. 
One can easily see that $\hat {\mathcal{I}}(T(\hat \Sigma_1),T(\hat \Sigma_2),T(\hat \Sigma_3))$ is $0$-approximate initial-state opaque, since for any run starting from any secret state, i.e. $z_3$ and $z_7$, there exists a run from a non-secret state, i.e. $z_1$ and $z_5$, such that the output trajectories are exactly the same. Essentially, one can verify that the abstract network holds this property regardless of the number of systems (i.e. n), due to the homogeneity of systems $\Sigma_i$ and the symmetry of the circular network topology. Thus, one can conclude that $T(\hat \Sigma) = \hat {\mathcal{I}}(T(\hat{\Sigma}_1),\dots,T(\hat{\Sigma}_n))$ is $0$-approximate initial-state opaque. Therefore, by Corollary \ref{thm:InitSOP}, we obtain that the original network $T(\Sigma) = \mathcal{I}(T(\Sigma_1),\dots,T(\Sigma_n))$ is $0.5$-approximate initial-state opaque.

\section{Conclusion} \label{sec:conclusion}
In this paper, we provided a compositional framework for the construction of opacity-preserving finite abstractions for networks of discrete-time switched systems. First, an approximate initial-state opacity-preserving simulation function (InitSOPSF) is defined to characterize the simulation relation between two networks, which facilitates the abstraction-based opacity verification process. Then we presented a compositional approach to construct finite abstractions locally for concrete subsystems under incremental input-to-state stability property. The interconnection of local finite abstractions forms an abstract network that mimics the behaviors of the concrete network while preserving opacity via the proposed InitSOPSF. Futhermore, we derived a small-gain type condition, under which one can guarantee the existence of proper quantization parameters for the construction of finite abstractions. For future work, we are interested in extending the compositionality results to cover more notions of opacity, e.g., current-state opacity \cite{saboori2007notions}, K-step opacity \cite{saboori2011verification}, and infinite-state opacity \cite{saboori2012verification}. Moreover, it would be an interesting direction to investigate opacity property for large-scale switched systems with unstable mode, and also for other classes of hybrid systems, e.g., stochastic systems and impulsive systems.

\bibliographystyle{IEEEtran}      
\bibliography{mybibfile} 	
\end{document}